\documentclass[onefignum,onetabnum]{siamart251216}
\usepackage{subfigure}    
\usepackage{epsfig,soul,comment,bm,booktabs,multirow}
\usepackage{lipsum}
\usepackage{amsfonts}
\usepackage{graphicx,epstopdf}
\usepackage{algorithmic}
\usepackage{amsopn}
 
\ifpdf
  \DeclareGraphicsExtensions{.eps,.pdf,.png,.jpg}
\else
  \DeclareGraphicsExtensions{.eps}
\fi

\soulregister{\ref}{7}
\soulregister{\cite}{7}

\crefname{hypothesis}{Hypothesis}{Hypotheses}

\newsiamremark{remark}{Remark}

\newtheorem{assumption}{Assumption}[section]

\headers{DD-ASD for Wetting Transitions}{T. Xu, L. Zhang, and Y. Zhang}

\title{Structure-preserving diffuse-domain accelerated saddle dynamics for wetting transitions\thanks{Submitted to the editors DATE.
\funding{L.~Zhang was supported by the National Natural Science Foundation of China (No.~12225102, T2321001, and 12288101) and the National Key Research and Development Program of China 2024YFA0919500. Y.~Zhang was supported by the National Natural Science Foundation of China (No.~12301554).}}}

\author{Tianye Xu\thanks{Ministry of Education Key Laboratory of NSLSCS, Nanjing Normal University, Nanjing, People's Republic of China (\email{250902012@njnu.edu.cn}).}
\and Lei Zhang\thanks{School of Mathematical Sciences, Beijing International Center for Mathematical Research, Center for Quantitative Biology, Center for Machine Learning Research, Institute for Artificial Intelligence, Peking University, Beijing 100871, China (\email{zhangl@math.pku.edu.cn}).}
\and Yuze Zhang\thanks{Ministry of Education Key Laboratory of NSLSCS, Nanjing Normal University, Nanjing, People's Republic of China (\email{05439@njnu.edu.cn}).}}

\begin{document}

\maketitle

\begin{abstract}
Wetting transitions on textured substrates play a central role in the design of functional surfaces, but resolving their transition mechanisms requires efficient exploration of complex energy landscapes. In particular, saddle points are essential for revealing the connectivity among metastable states and transition pathways. In this work, we develop a diffuse-domain accelerated saddle dynamics (DD-ASD) framework for mass-constrained wetting transitions on complex textured substrates. 
The diffuse-domain formulation embeds complex solid geometries into a Cartesian grid, avoiding body-fitted mesh construction in repeated saddle-point searches. 
An orthogonal projection is applied consistently to the phase-field state and unstable-direction dynamics, preserving the prescribed droplet mass at the discrete level. 
Momentum acceleration is incorporated into the saddle dynamics to improve the efficiency of high-index saddle searches. 
Under the stated local spectral and exact-eigenspace assumptions, we establish a local convergence rate of $1-\mathcal{O}(1/\sqrt{\kappa})$ on the effective mass-conserving subspace. Numerical experiments demonstrate the effectiveness of the proposed method in resolving wetting solution landscapes, transition pathways, and energy barriers on textured substrates. We further extend the framework to fully three-dimensional wetting-landscape computations.
\end{abstract}



\section{Introduction}
Wetting transitions on textured substrates are central to the design of functional surfaces, including self-cleaning, anti-icing, water-harvesting, and directional-transport materials \cite{barthlott1997purity, quere2008wetting, kreder2016design, park2016condensation, zheng2010directional, feng2021capillary}. A droplet on a rough surface may exhibit several metastable morphologies between the Cassie-Baxter and Wenzel states \cite{cassie1944contact, wenzel1936resistance}. Understanding how these states are connected requires resolving not only local minimizers but also saddle points on the underlying wetting energy landscape \cite{bormashenko2015progress, koishi2009morphological, ren2014wetting}.

Phase-field models provide a convenient variational framework for wetting problems because topological changes in the liquid interface can be represented without explicit tracking of the interface \cite{cahn1958free, qian2003molecular, xu2011phase}. In such models, wetting transitions can be understood as minimum energy paths (MEPs). An index-1 saddle point characterizes a transition state connecting two local minima along a MEP, while high-index saddle points connect lower-index saddle points and thereby organize the hierarchical structure of the solution landscape \cite{milnor1963morse, yin2020PRL}. High-index saddle dynamics (HiSD) \cite{yin2019high, yin2021searching} has therefore become a useful tool for computing saddle points of prescribed Morse indices and constructing solution landscapes for complex systems including liquid crystals \cite{shi2023hierarchies, wu2026prl}, block copolymers \cite{zhang2026prr}, Bose--Einstein condensates \cite{yin2023revealing}, and jammed granular matter \cite{wu2026jamming}.

For wetting transitions on complex textures, Zhang et al.~\cite{zhang2025solution} constructed two-dimensional (2D) droplet landscapes using HiSD coupled with a body-fitted finite-element discretization. 
That study resolved several important transition mechanisms, but the use of unstructured meshes and assembled finite-element operators makes repeated saddle-point computations increasingly costly as the geometry becomes more complex and the problem size increases. 
In addition, the previous implementation required separate treatment of the prescribed volume constraint. 
These costs become more significant in systematic parameter studies and three-dimensional (3D) calculations.

The diffuse-domain method and related domain-embedding methods provide an alternative way to handle complex geometries by representing the physical domain within a simple computational region~\cite{glowinski1994fictitious,qian2025}. 
This allows Cartesian-grid discretizations to be used without constructing a body-fitted mesh for each substrate geometry. 
For the present problem, the main difficulty is to incorporate this geometric representation into a mass-constrained saddle-point search while retaining the structure and efficiency required for solution-landscape computations.

This work develops a diffuse-domain accelerated saddle dynamics (DD-ASD) method for computing wetting transitions and energy barriers on complex textured substrates. 
The diffuse-domain formulation represents the solid geometry on a Cartesian grid, while the mass constraint is enforced by an orthogonal projection applied to both the state and unstable-direction dynamics. 
Momentum acceleration is then incorporated into the saddle search to reduce the cost associated with ill-conditioned phase-field systems. 
The aim is to provide a practical framework for systematic transition-pathway and energy-barrier calculations without relying on body-fitted discretizations for each complex substrate geometry.

The main contributions are as follows:
\begin{itemize}
    \item We develop a diffuse-domain accelerated solution-landscape framework for mass constrained wetting transitions on complex textured substrates. 
    The Cartesian-grid formulation avoids body-fitted mesh construction in repeated saddle-point calculations.

    \item We apply the mass-conserving projection consistently to both the phase-field state and unstable-direction dynamics. 
    The resulting discrete iteration preserves the prescribed mass in exact arithmetic and to machine precision in computation.

    \item Under the stated local spectral and exact-eigenspace assumptions, we establish a local convergence rate of $1-\mathcal{O}(1/\sqrt{\kappa})$ for the momentum-accelerated state iteration in the effective mass-conserving subspace.
    \item We verify the numerical reliability of DD-ASD in 2D through step-size self-convergence, mass conservation, and mesh sensitivity. The computational efficiency of the framework makes fully 3D wetting-landscape calculations practical.

\end{itemize}

The rest of the paper is organized as follows. Section 2 introduces the phase-field energy functional in the diffuse domain and formulates the mass-conserving DD-ASD algorithm. Section 3 establishes the spectral structure analysis and the Lipschitz continuity of the eigen-subspace, and presents the local convergence theory and the proof of the accelerated rate for the discrete scheme. Section 4 provides numerical experiments, including accuracy validation, efficiency comparisons with FEM, benchmark constructions of solution landscapes in both 2D and 3D. Concluding remarks are given in Section 5.

\section{Diffuse-domain accelerated saddle dynamics for wetting transition}

\subsection{Diffuse-domain formulation of wetting energy}

For a physical fluid domain $\Omega_f$ with a solid boundary $\partial\Omega_f$, the standard Ginzburg-Landau free energy functional of the phase-field system with wetting effects is written as:
\begin{equation}
\mathcal{E}_0[u] = \int_{\Omega_f} \left[ \frac{\epsilon}{2} |\nabla u|^2 + \frac{1}{\epsilon} F(u) \right] d\mathbf{x} 
+ \int_{\partial\Omega_f} W(u) \, dS,
\end{equation}
where the functions $F(u)$ and $W(u)$ are defined as:
\begin{equation}
F(u) = \frac{1}{4}(u^2 - 1)^2, \quad W(u) = -\frac{\sqrt{2}}{2}\cos\theta\left(u - \frac{u^3}{3}\right).
\end{equation}


%
Following the domain-embedding framework
in~\cite{qian2025}, we embed the physical fluid domain $\Omega_f$
into a regular rectangular computational domain $\Omega$.
The geometry can be described by the characteristic function of
$\Omega_f$. For the continuous variational formulation, we use
its smooth approximation
\begin{equation}
\label{eq:psi_smooth}
\psi_{\varepsilon_g}(\mathbf{x})
=
\frac{1}{2}
\left(
1-\tanh\!\left(\frac{3r(\mathbf{x})}{\varepsilon_g}\right)
\right),
\end{equation}
where $r(\mathbf{x})$ denotes the signed distance to the solid boundary $\partial\Omega_f$. We take $r<0$ in $\Omega_f$ and $r>0$ outside $\overline{\Omega_f}$.
The parameter $\varepsilon_g>0$ controls the width of the diffuse geometric transition. As $\varepsilon_g\to0$, $\psi_{\varepsilon_g}$ approaches the characteristic function of $\Omega_f$ away from the boundary, while $|\nabla\psi_{\varepsilon_g}|$ provides a smooth approximation of the surface delta function. This allows the bulk and boundary contributions to be written as integrals over $\Omega$ in the continuous energy. Throughout the continuous formulation, $\psi_{\varepsilon_g}$ is assumed to have the regularity required for the variations below. For notational simplicity, we write $\psi=\psi_{\varepsilon_g}$.

In the numerical discretization, we use the characteristic function directly to represent the geometry on the Cartesian grid:
\begin{equation}
\psi_b(\mathbf{x})
=
\begin{cases}
1, & \mathbf{x}\in\Omega_f,\\
0, & \mathbf{x}\in\Omega\setminus\Omega_f.
\end{cases}
\end{equation}
This binary mask assigns fluid and solid grid points the limiting values of the smooth indicator, without introducing an additional diffuse geometric layer. It is used in the discrete bulk terms and mass constraint, while the boundary contribution is discretized separately below. The discrete energy and the subsequent iteration analysis are defined directly using this representation and do not require the binary mask to be differentiable. Thus, $\psi$ denotes the smooth indicator used in the continuous variational formulation, whereas $\psi_b$ denotes the binary mask used in the discrete computations.

Before imposing the mass constraint, we derive the unconstrained gradient and the Hessian of $\mathcal{E}$.

\noindent\textbf{Unprojected gradient.}
For an arbitrary test direction $v$, integration by parts under the periodic boundary conditions on $\Omega$, for which the boundary contributions cancel, gives
\begin{equation}
D\mathcal{E}(u)[v]
=
\int_{\Omega}
\left[
-\epsilon\nabla\cdot(\psi\nabla u)
+\frac{1}{\epsilon}\psi F'(u)
+|\nabla\psi|W'(u)
\right]v\,d\mathbf{x}.
\end{equation}
Therefore, the unprojected $L^2$ gradient is
\begin{equation}
\label{eq:grad_unproj}
\nabla\mathcal{E}(u)
=
-\epsilon\nabla\cdot(\psi\nabla u)
+\frac{1}{\epsilon}\psi F'(u)
+|\nabla\psi|W'(u).
\end{equation}

\noindent\textbf{Unprojected Hessian.}
Differentiating \eqref{eq:grad_unproj} in the direction $v$ gives
\begin{equation}
\label{eq:hess_unproj}
\mathcal{H}_0(u)v
=
-\epsilon\nabla\cdot(\psi\nabla v)
+\frac{1}{\epsilon}\psi F''(u)v
+|\nabla\psi|W''(u)v,
\end{equation}
where
\[
F''(u)=3u^2-1,
\qquad
W''(u)=\sqrt{2}\cos\theta\,u.
\]

\subsubsection{Intrinsic Mass Conservation and Invariant Subspaces}

For the mass-constrained wetting model considered here, mass conservation is imposed in the diffuse-domain formulation through
\begin{equation}
M_{\psi}(u(t)):=\int_{\Omega}u(\mathbf{x},t)\psi(\mathbf{x})\,d\mathbf{x}=M_0.
\end{equation}
Since the fixed geometry implies $\psi_t=0$, differentiation gives
\begin{equation}
\int_{\Omega}u_t\psi\,d\mathbf{x}
=
\langle u_t,\psi\rangle
=
0.
\end{equation}
Thus, $u_t$ is orthogonal to the mass direction $\psi$. For the binary mask $\psi_b=\chi_{\Omega_f}$, this weighted mass coincides with the physical-domain mass:
\[
\int_{\Omega}u\psi_b\,d\mathbf{x}
=
\int_{\Omega_f}u\,d\mathbf{x}.
\]

To enforce this constraint, we introduce a spatially independent Lagrange multiplier $\lambda(t)$ into the standard unprojected gradient flow:
\begin{equation}
u_t = -\nabla\mathcal{E}(u) + \lambda \psi.
\end{equation}

By taking the inner product of both sides with $\psi$ and enforcing the constraint $\langle u_t, \psi \rangle = 0$, the exact Lagrange multiplier is analytically determined as:
\begin{equation}
0 = -\langle \nabla\mathcal{E}(u), \psi \rangle + \lambda \langle \psi, \psi \rangle \quad \Longrightarrow \quad \lambda = \frac{\langle \nabla\mathcal{E}(u), \psi \rangle}{\|\psi\|^2}.
\end{equation}

Substituting $\lambda$ back into the gradient flow gives the mass-conserving dynamics:
\begin{equation} \label{eq:gradient}
u_t = -\left(\nabla\mathcal{E}(u) - \frac{\langle \nabla\mathcal{E}(u), \psi \rangle}{\|\psi\|^2}\psi \right) =: -\nabla_{\mathcal{P}}\mathcal{E}(u).
\end{equation}

This operation naturally gives rise to the continuous orthogonal projection operator $\mathcal{P}: L^2(\Omega) \to \mathcal{V}^c_{\text{eff}}$, where the effective mass-conserving subspace is defined by
\begin{equation}
\mathcal{V}^c_{\text{eff}} := \left\{ v \in L^2(\Omega) \;\middle|\; \int_{\Omega} v(\mathbf{x})\psi(\mathbf{x})\,d\mathbf{x} = 0 \right\}.
\end{equation}
The action of $\mathcal{P}$ on any test function $v$ is explicitly given by
\begin{equation}
\mathcal{P} v = v - \frac{\langle v, \psi \rangle}{\|\psi\|^2}\psi, \quad \text{where } \langle v, \psi \rangle = \int_{\Omega} v(\mathbf{x})\psi(\mathbf{x}) d\mathbf{x}.
\end{equation}

Consequently, the constrained gradient is $\nabla_{\mathcal P}\mathcal E(u)=\mathcal P\nabla\mathcal E(u)$. The component removed by $\mathcal P$ is the Lagrange-multiplier contribution in the mass direction. Because the multiplier is evaluated explicitly, no additional coupled saddle-point system is required.

Since $\mathcal P$ is independent of $u$, the Hessian restricted to the
mass-conserving subspace is
\begin{equation}
\label{eq:hess_proj}
\nabla_{\mathcal P}^2\mathcal E(u)
=
\mathcal P\mathcal H_0(u)\mathcal P.
\end{equation}

In the continuous-time and continuous-space setting, we formulate the mass-conserving Diffuse-Domain Saddle Dynamics system for computing index-$k$ saddle points. Let $u(\mathbf{x}, t) \in H^1(\Omega)$ be the continuous phase-field state, and let $\{v_i(\mathbf{x}, t)\}_{i=1}^k \subset H^1(\Omega)$ be the $k$ orthonormal unstable search directions. The coupled continuous dynamical system is written as:
\begin{equation} \label{eq:continuous_ph_sd}
\begin{cases}
\partial_t u = -\mathcal{P} \nabla\mathcal{E}(u) + 2 \sum_{i=1}^k \langle \mathcal{P} \nabla\mathcal{E}(u), v_i \rangle v_i, \\
\partial_t v_i = -\mathcal{P} \mathcal{H}_0(u) v_i + \langle \mathcal{P} \mathcal{H}_0(u) v_i, v_i \rangle v_i + 2 \sum_{j=1}^{i-1} \langle \mathcal{P} \mathcal{H}_0(u) v_i, v_j \rangle v_j, \quad i = 1, \dots, k,
\end{cases}
\end{equation}
where $\langle \cdot, \cdot \rangle$ denotes the standard $L^2(\Omega)$ inner product, $\mathcal{P}$ is the continuous projection operator, and $v_i$ satisfy the continuous orthonormality conditions $\langle v_i, v_j \rangle = \delta_{ij}$ and the mass-conserving constraints $\langle v_i, \psi \rangle = 0$ for all $1 \le i, j \le k$.

\subsection{Diffuse-domain accelerated saddle dynamics}
Before presenting the discrete iteration, we define the discrete energy and its associated operators. On a uniform Cartesian grid with spacing $h$ in spatial dimension $d\in\{2,3\}$, let $N$ denote the total number of grid points. The phase field and binary fluid-domain mask are represented by $\mathbf u,\bm\psi_b\in\mathbb R^N$, respectively. The
discrete energy is
\begin{equation}
\label{eq:discrete_energy}
\mathcal E_h(\mathbf u)
=
h^d\left[
\frac{\epsilon}{2}\mathbf u^\top\mathbf L_{\psi,h}\mathbf u
+\frac{1}{\epsilon}\bm\psi_b^\top\mathbf F(\mathbf u)
+\mathbf d^\top\mathbf W(\mathbf u)
\right],
\end{equation}
where $\mathbf L_{\psi,h}$ is the discrete negative diffusion operator. $\mathbf F$ and $\mathbf W$ are evaluated componentwise, and $\mathbf d\in\mathbb R^N$ contains the discrete solid--fluid interface weights.

To specify $\mathbf d$ in two dimensions, let $\psi_{i,j}\in\{0,1\}$ denote the binary mask at $(x_j,y_i)$, where the first and second indices correspond to the $y$- and $x$-directions, respectively.  Let $D=(d_{i,j})$ denote the boundary-weight matrix, and set
$\mathbf d=\operatorname{vec}(D)$. Define the fluid-side interface
indicators by
\[
b_{i,j}=\psi_{i,j}(1-\psi_{i-1,j}),\qquad
t_{i,j}=\psi_{i,j}(1-\psi_{i+1,j}),
\]
\[
\ell_{i,j}=\psi_{i,j}(1-\psi_{i,j-1}),\qquad
r_{i,j}=\psi_{i,j}(1-\psi_{i,j+1}).
\]
The entries of $D$ are defined by

\begin{equation}
d_{i,j}
=
\frac{b_{i,j}+t_{i,j}}{h}
+
\frac{(1-b_{i,j})(\ell_{i,j}+r_{i,j})}{h}.
\end{equation}
Here, mask values outside the computational grid are set to zero. The factor $1-b_{i,j}$ implements the corner convention used in our discretization: nodes adjacent to solid material below receive no additional left or right contribution.

We construct $\mathbf L_{\psi,h}$ using paired forward and backward differences under periodic boundary conditions. In two dimensions, the forward differences are
\begin{equation}\label{eq:forward_diff}
\partial_x^+u_{i,j}
=\frac{u_{i,j+1}-u_{i,j}}{h},
\qquad
\partial_y^+u_{i,j}
=\frac{u_{i+1,j}-u_{i,j}}{h},
\end{equation}
with fluxes
\[
q_{x,i,j}=(\psi_b)_{i,j}\partial_x^+u_{i,j},
\qquad
q_{y,i,j}=(\psi_b)_{i,j}\partial_y^+u_{i,j}.
\]
The corresponding backward differences are
\begin{equation}\label{eq:backward_diff}
\partial_x^-q_{x,i,j}
=\frac{q_{x,i,j}-q_{x,i,j-1}}{h},
\qquad
\partial_y^-q_{y,i,j}
=\frac{q_{y,i,j}-q_{y,i-1,j}}{h},
\end{equation}
and the discrete negative diffusion operator is
\[
(\mathbf L_{\psi,h}\mathbf u)_{i,j}
=
-\partial_x^-q_{x,i,j}
-\partial_y^-q_{y,i,j}.
\]
All indices are interpreted periodically. Equivalently, in matrix form,
\[
\mathbf L_{\psi,h}
=
\sum_{\ell=1}^{d}
(\mathbf D_\ell^+)^\top
\operatorname{diag}(\bm\psi_b)\mathbf D_\ell^+,
\]
where $\mathbf D_\ell^+$ is the periodic forward-difference matrix in coordinate direction $\ell$. Thus, $\mathbf L_{\psi,h}$ is symmetric positive semidefinite. The 3D construction includes the corresponding differences in the third coordinate.

We define the discrete gradient with respect to
$\langle\mathbf a,\mathbf b\rangle_h:=h^d\mathbf a^\top\mathbf b$, so that
\[
D\mathcal E_h(\mathbf u)[\mathbf v]
=
\langle\nabla_h\mathcal E_h(\mathbf u),\mathbf v\rangle_h.
\]
The resulting gradient and Hessian are
\[
\begin{aligned}
\nabla_h\mathcal E_h(\mathbf u)
&=
\epsilon\mathbf L_{\psi,h}\mathbf u
+\frac{1}{\epsilon}\bm\psi_b\odot\mathbf F'(\mathbf u)
+\mathbf d\odot\mathbf W'(\mathbf u),\\
\nabla_h^2\mathcal E_h(\mathbf u)
&=
\epsilon\mathbf L_{\psi,h}
+\operatorname{diag}\left(
\frac{1}{\epsilon}\bm\psi_b\odot\mathbf F''(\mathbf u)
+\mathbf d\odot\mathbf W''(\mathbf u)
\right),
\end{aligned}
\]
where $\odot$ denotes componentwise multiplication.

For the discrete mass constraint $h^d\mathbf u^\top\bm\psi_b=M_0$, the corresponding mass conserving projection matrix is
\[
\mathbf P_0
=
\mathbf I
-
\frac{\bm\psi_b\bm\psi_b^\top}
{\bm\psi_b^\top\bm\psi_b}.
\]
Since the grid weight $h^d$ is uniform, $\mathbf P_0$ is orthogonal with respect to both the Euclidean and discrete $L^2$ inner products. The projected gradient and Hessian are
\begin{equation}\label{eq:discrete_grad_proj}
\begin{aligned}
\nabla_{\mathbf P}\mathcal E_h(\mathbf u)
&=\mathbf P_0\nabla_h\mathcal E_h(\mathbf u),\\
\nabla_{\mathbf P}^2\mathcal E_h(\mathbf u)
&=\mathbf P_0\nabla_h^2\mathcal E_h(\mathbf u)\mathbf P_0.
\end{aligned}
\end{equation}

Guided by the projected structure of \eqref{eq:continuous_ph_sd},
we define the fully discrete DD-ASD iteration using the binary-mask
discretization introduced above. To discretize \eqref{eq:continuous_ph_sd}, we use an explicit Euler scheme and add the heavy-ball momentum term $\gamma(\mathbf u^n-\mathbf u^{n-1})$ to accelerate the state iteration. The complete discrete scheme is
\begin{equation}
\label{eq:discrete_coupled_system}
\begin{cases}
\displaystyle
\mathbf{u}^{n+1}
=
\mathbf{u}^{n}
+\beta\mathbf{g}^{n}
+\gamma(\mathbf{u}^{n}-\mathbf{u}^{n-1}),
\\[1ex]
\displaystyle
\widehat{\mathbf{v}}_i^{\,n+1}
=
\mathbf{v}_i^n+\beta\mathbf{g}_i^n,
\qquad i=1,\ldots,k,
\\[1ex]
\displaystyle
\widetilde{\mathbf{v}}_i^{\,n+1}
=
\widehat{\mathbf{v}}_i^{\,n+1}
-
\sum_{j=1}^{i-1}
\left(
(\mathbf{v}_j^{n+1})^\top
\widehat{\mathbf{v}}_i^{\,n+1}
\right)
\mathbf{v}_j^{n+1},
\qquad i=1,\ldots,k,
\\[1ex]
\displaystyle
\mathbf{v}_i^{n+1}
=
\frac{\widetilde{\mathbf{v}}_i^{\,n+1}}
{\left\|\widetilde{\mathbf{v}}_i^{\,n+1}\right\|},
\qquad i=1,\ldots,k,
\\[2ex]
\displaystyle
l^{n+1}
=
\max\left\{
\frac{l^n}{1+\beta},
\epsilon_0
\right\}.
\end{cases}
\end{equation}
The direction vectors are updated sequentially in increasing order of $i$, with the sum in the third equation understood to be zero when $i=1$. Here, $\beta>0$ is the step size and $\gamma$ is the momentum parameter. We choose $l^0\geq\epsilon_0>0$. Unless otherwise stated, vector norms and direction orthonormality are Euclidean; $\|\cdot\|_2$ denotes the Euclidean vector norm or its induced matrix norm. In the numerical experiments, we write $\Delta t=\beta$.

The discrete update directions $\mathbf g^n$ and $\mathbf g_i^n$ are defined as follows:
\begin{equation} \label{eq:gn_discrete}
\mathbf{g}^n = -\mathbf{P}_0 \nabla_h \mathcal{E}_h(\mathbf{u}^n) + 2 \sum_{i=1}^k \left( (\mathbf{v}_i^n)^\top \mathbf{P}_0 \nabla_h \mathcal{E}_h(\mathbf{u}^n) \right) \mathbf{v}_i^n,
\end{equation}
\begin{equation} \label{eq:gin_discrete}
\mathbf{g}_i^n = -\mathbf{P}_0 \mathbf{H}^n_i + \left( (\mathbf{v}_i^n)^\top \mathbf{P}_0 \mathbf{H}^n_i \right) \mathbf{v}_i^n + 2 \sum_{j=1}^{i-1} \left( (\mathbf{v}_j^n)^\top \mathbf{P}_0 \mathbf{H}^n_i \right) \mathbf{v}_j^n.
\end{equation}
Here,
\[
\mathbf H_i^n:=\mathbf H(\mathbf u^{n+1},\mathbf v_i^n,l^n)
\]
is the centered finite-difference approximation of the projected Hessian--vector product:
 \begin{equation}
\begin{aligned}\label{eq:hess_fd_discrete}
\mathbf{H}(\mathbf{u}^{n+1}, \mathbf{v}_i^n, l^n) 
&= \frac{\mathbf{P}_0\nabla_h\mathcal{E}_h(\mathbf{u}^{n+1} + l^n\mathbf{v}_i^n) 
      - \mathbf{P}_0\nabla_h\mathcal{E}_h(\mathbf{u}^{n+1} - l^n\mathbf{v}_i^n)}{2l^n}\\
&\approx \nabla^2_{\mathbf{P}}\mathcal{E}_h(\mathbf{u}^{n+1})\,\mathbf{v}_i^n.
\end{aligned}
\end{equation}
Since the orthonormal eigenvectors $\mathbf{v}_i^n$ are strictly maintained within the discrete mass-conserving subspace satisfying $(\mathbf{v}_i^n)^\top\bm\psi_b = 0$, and because the Hessian-vector product $\mathbf{H}^n_i$ is already projected ($\mathbf{P}_0 \mathbf{H}^n_i = \mathbf{H}^n_i$), the update directions in \eqref{eq:gn_discrete} and \eqref{eq:gin_discrete} can be simplified into the following compact matrix-vector forms:
\begin{equation}
\mathbf{g}^n = -\left( \mathbf{I} - 2 \sum_{i=1}^k \mathbf{v}_i^n (\mathbf{v}_i^n)^\top \right) \nabla_\mathbf{P} \mathcal{E}_h(\mathbf{u}^n),
\end{equation}
\begin{equation}
\mathbf{g}_i^n = -\left( \mathbf{I} - \mathbf{v}_i^n (\mathbf{v}_i^n)^\top - 2 \sum_{j=1}^{i-1} \mathbf{v}_j^n (\mathbf{v}_j^n)^\top \right) \mathbf{H}(\mathbf{u}^{n+1}, \mathbf{v}_i^n, l^n),
\end{equation}

Starting from $\mathbf{u}^{-1}=\mathbf{u}^0$ and a set of orthonormal
directions $\{\mathbf{v}_i^0\}_{i=1}^k$ satisfying
\[
(\mathbf{v}_i^0)^\top\mathbf{v}_j^0=\delta_{ij},
\qquad
(\mathbf{v}_i^0)^\top\bm{\psi}_b=0,
\]
we apply \eqref{eq:discrete_coupled_system} for $n=0,1,\ldots$ until
\[
\left\|
\nabla_{\mathbf P}\mathcal{E}_h(\mathbf{u}^n)
\right\|
<
\epsilon_{\mathrm{tol}}.
\]

\subsection{Mass Conservation and the Constrained Subspace}

On a uniform Cartesian grid with spacing $h$ in spatial dimension $d\in\{2,3\}$, define the discrete phase-field mass by
\[
M_h(\mathbf u)
:=
h^d\mathbf u^\top\bm\psi_b,
\qquad
M_0:=M_h(\mathbf u^0),
\]
where $\bm\psi_b\in\mathbb R^N$ is the fixed binary mask of the fluid domain. The mass constraint is therefore $M_h(\mathbf u)=M_0$. Its associated tangent subspace is
\begin{equation}
\label{eq:effective_subspace}
\mathcal V_{\mathrm{eff}}
=
\operatorname{span}\{\bm\psi_b\}^{\perp}
=
\left\{
\mathbf w\in\mathbb R^N
\;\middle|\;
\mathbf w^\top\bm\psi_b=0
\right\}.
\end{equation}

By the definition of the discrete projection operator $\mathbf{P}_0 = \mathbf{I} - \frac{\bm{\psi}_b\bm{\psi}_b^\top}{\bm{\psi}_b^\top\bm{\psi}_b}$, we have $\mathbf{P}_0 \bm{\psi}_b = \mathbf{0}$. Thus, for any discrete state vector $\mathbf{u}$, the projected Hessian matrix satisfies:
\begin{equation}
\nabla^2_{\mathbf{P}} \mathcal{E}_h(\mathbf{u}) \bm{\psi}_b = \mathbf P_0\nabla_h^2\mathcal E_h(\mathbf u) (\mathbf{P}_0 \bm{\psi}_b) = \mathbf{0}.
\end{equation}
Thus, $\bm\psi_b$ is a zero eigenvector of the projected Hessian. The following theorem shows that the state increments and search directions remain in $\mathcal V_{\mathrm{eff}}$, so the state iterates preserve the prescribed mass.
\begin{theorem}[Discrete mass conservation]
\label{thm:mass_conservation}
Let $d\in\{2,3\}$ denote the spatial dimension and let $\bm\psi_b\neq\mathbf 0$ be fixed. On a uniform Cartesian grid with spacing $h$, define the discrete mass by
\[
M_h(\mathbf u):=h^d\mathbf u^\top\bm\psi_b.
\]
Suppose that $\mathbf u^{-1}=\mathbf u^0$, $M_h(\mathbf u^0)=M_0$, and the initial directions satisfy
\[
(\mathbf v_i^0)^\top\mathbf v_j^0=\delta_{ij},
\qquad
(\mathbf v_i^0)^\top\bm\psi_b=0,
\qquad 1\leq i,j\leq k.
\]
Assume that all orthonormalization steps are well defined. Then, in exact arithmetic, the recurrence in \eqref{eq:discrete_coupled_system} preserves
\[
M_h(\mathbf u^n)=M_0,
\qquad
(\mathbf v_i^n)^\top\bm\psi_b=0,
\qquad n\geq0,\quad 1\leq i\leq k.
\]
\end{theorem}

\begin{proof}
Since $\mathbf P_0^\top=\mathbf P_0$ and $\mathbf P_0\bm\psi_b=\mathbf 0$, the projected gradient satisfies
\begin{equation}\label{eq:grad_orth}
\bm\psi_b^\top\nabla_{\mathbf P}\mathcal E_h(\mathbf u)
=
\bm\psi_b^\top\mathbf P_0\nabla_h\mathcal E_h(\mathbf u)
=0.
\end{equation}
Likewise, for any $\mathbf w$,
\begin{equation}\label{eq:hess_orth}
\bm\psi_b^\top\nabla_{\mathbf P}^2\mathcal E_h(\mathbf u)\mathbf w
=
\bm\psi_b^\top\mathbf P_0\nabla_h^2\mathcal E_h(\mathbf u)\mathbf P_0\mathbf w
=0.
\end{equation}
The finite-difference vector $\mathbf H_i^n$ in~\eqref{eq:hess_fd_discrete} is a difference of projected gradients and therefore also satisfies $\bm\psi_b^\top\mathbf H_i^n=0$.

We first prove invariance of the directions by induction. The claim holds at $n=0$. Suppose $\bm\psi_b^\top\mathbf v_i^n=0$ for all $i$. Taking the inner product of~\eqref{eq:gin_discrete} with $\bm\psi_b$ gives
\[
\begin{aligned}
\bm\psi_b^\top\mathbf g_i^n
={}&-\bm\psi_b^\top\mathbf H_i^n
+\bigl((\mathbf v_i^n)^\top\mathbf H_i^n\bigr)
  \bm\psi_b^\top\mathbf v_i^n\\
&+2\sum_{j=1}^{i-1}
\bigl((\mathbf v_j^n)^\top\mathbf H_i^n\bigr)
\bm\psi_b^\top\mathbf v_j^n
=0.
\end{aligned}
\]
Hence the tentative update
\[
\widehat{\mathbf v}_i^{\,n+1}
=
\mathbf v_i^n+\beta\mathbf g_i^n
\]
remains orthogonal to $\bm\psi_b$. The Gram--Schmidt step and subsequent normalization are
\[
\widetilde{\mathbf v}_i^{\,n+1}
=
\widehat{\mathbf v}_i^{\,n+1}
-
\sum_{j=1}^{i-1}
\left(
(\mathbf v_j^{n+1})^\top
\widehat{\mathbf v}_i^{\,n+1}
\right)
\mathbf v_j^{n+1},
\qquad
\mathbf v_i^{n+1}
=
\frac{\widetilde{\mathbf v}_i^{\,n+1}}
{\|\widetilde{\mathbf v}_i^{\,n+1}\|_2}.
\]
For $i=1$, the sum is understood to be zero. Assuming that
$\bm\psi_b^\top\mathbf v_j^{n+1}=0$ for $j<i$, induction over $i$ gives
\[
\bm\psi_b^\top\widetilde{\mathbf v}_i^{\,n+1}=0
\qquad\text{and hence}\qquad
\bm\psi_b^\top\mathbf v_i^{n+1}=0.
\]
Thus, all updated direction vectors remain in the mass-conserving subspace, which completes the induction over $n$.

Using this orthogonality and~\eqref{eq:grad_orth}, the state direction in~\eqref{eq:gn_discrete} satisfies
\[
\begin{aligned}
\bm\psi_b^\top\mathbf g^n
={}&-\bm\psi_b^\top\nabla_{\mathbf P}\mathcal E_h(\mathbf u^n)\\
&+2\sum_{i=1}^k
\bigl((\mathbf v_i^n)^\top\nabla_{\mathbf P}\mathcal E_h(\mathbf u^n)\bigr)
\bm\psi_b^\top\mathbf v_i^n
=0.
\end{aligned}
\]
Set $m_n:=M_h(\mathbf u^n)$. Taking the discrete mass of the state update yields
\[
m_{n+1}-m_n
=
\beta h^d\bm\psi_b^\top\mathbf g^n
+\gamma(m_n-m_{n-1})
=
\gamma(m_n-m_{n-1}).
\]
Since $m_{-1}=m_0=M_0$, induction gives $m_n=M_0$ for every $n\geq0$.
\end{proof}
\vspace{1em}
\section{Theoretical Convergence Analysis of DD-ASD}

In this section, we analyze the local convergence of the discrete DD-ASD iteration. We first establish continuity of the unstable-subspace projector, then derive estimates for the linearized iteration, and finally prove local convergence of the nonlinear iteration.



The full-grid discretization contains two types of points satisfying $(\bm{\psi}_b)_j=0$: completely inactive exterior points and auxiliary exterior points involved through the finite-difference discretization. Although an auxiliary exterior degree of freedom need not itself define a zero mode, both types of exterior degrees of freedom may support nonphysical mask-induced null modes.

Let $\mathcal{Z}_{\mathrm{ext}}\subset\mathbb{R}^N$ denote the subspace spanned by all such zero modes supported entirely on the exterior degrees of freedom. We assume that $\mathcal{Z}_{\mathrm{ext}}$ is independent of $\mathbf{u}$ in the neighborhood considered below, and perform the convergence analysis on the reduced space
\[
\mathcal{X}_{\mathrm{red}}
=
\mathcal{Z}_{\mathrm{ext}}^\perp,
\]
or, equivalently, on the quotient space $\mathbb{R}^N/\mathcal{Z}_{\mathrm{ext}}$. This reduction removes only the nonphysical exterior null modes and retains the zero mode generated by the mass-conservation projection.

Let $N_{\mathrm{in}}=\dim(\mathcal{X}_{\mathrm{red}})$, and let $\mathbf{R}\in\mathbb{R}^{N\times N_{\mathrm{in}}}$ have orthonormal columns spanning $\mathcal{X}_{\mathrm{red}}$. Since the exterior null modes are supported where $\bm{\psi}_b=0$, the mass direction $\bm{\psi}_b$ is orthogonal to $\mathcal{Z}_{\mathrm{ext}}$. After an orthogonal change of coordinates in $\mathcal{X}_{\mathrm{red}}$, it may therefore be represented, up to an irrelevant nonzero scaling factor, by
\[
\mathbf{e}
=
[1,1,\ldots,1]^\top
\in\mathbb{R}^{N_{\mathrm{in}}}.
\]
Accordingly, the effective mass-conserving tangent subspace is
\begin{equation}
\label{eq:restricted_veff}
\widehat{\mathcal{V}}_{\mathrm{eff}}
=
\left\{
\mathbf{w}\in\mathbb{R}^{N_{\mathrm{in}}}
\;\middle|\;
\mathbf{w}^{\top}\mathbf{e}=0
\right\},
\end{equation}
which has dimension $N_{\mathrm{in}}-1$. For notational simplicity, we retain the symbols $\mathbf{u}$ and $\nabla_{\mathbf{P}}^2\mathcal{E}_h(\mathbf{u})$ for the reduced coordinates and the corresponding reduced operator in the remainder of this section.
\subsection{Spectral Properties and Subspace Continuity}

To facilitate the local convergence analysis near a target saddle point, we introduce the following standard assumptions regarding the spectral properties of the projected Hessian operator.
\begin{assumption}\label{assump:1}
Let $\mathbf u^*$ be an index-$k$ saddle point, and let $U(\mathbf u^*,\delta)= \{\mathbf u:\|\mathbf u-\mathbf u^*\|_2<\delta\}$. Assume that $\mathbf u^0\in U(\mathbf u^*,\delta)$ and that the projected Hessian satisfies the following conditions throughout this neighborhood:
\begin{enumerate}
    \item \textbf{Lipschitz Continuity:} There exists a constant $M > 0$ such that for any $\mathbf{u},\mathbf{w}\in U(\mathbf{u}^*,\delta)$,
    \[
    \|\nabla^2_{\mathbf{P}} \mathcal{E}_h(\mathbf{u}) - \nabla^2_{\mathbf{P}} \mathcal{E}_h(\mathbf{w})\|_2 \leq M\|\mathbf{u} - \mathbf{w}\|_2.
    \]
    \item \textbf{Null Space and Spectral Structure:}
Under the reduction of the exterior null modes described above, for any $\mathbf{u}\in U(\mathbf{u}^*,\delta)$, the reduced projected Hessian possesses a simple zero eigenvalue corresponding to the mass-conservation direction $\mathbf{e}$, and its eigenvalues satisfy
\[
\lambda_1
\leq \cdots \leq
\lambda_k
<0
=
\lambda_{k+1}
<
\lambda_{k+2}
\leq \cdots \leq
\lambda_{N_{\mathrm{in}}}.
\]
Moreover, there exist constants $0<\mu< L$ such that, uniformly for
all $\mathbf{u}\in U(\mathbf{u}^*,\delta)$,
\[
\mu
\leq
|\lambda_i(\mathbf{u})|
\leq
L,
\qquad
\lambda_i(\mathbf{u})\neq 0.
\]
\end{enumerate}
\end{assumption}

The following result follows from the Davis--Kahan theorem~\cite[Theorem~2]{yu2015useful} and gives a local Lipschitz bound for the unstable subspace projector.

\begin{corollary}[Unstable subspace continuity]
\label{cor:1}
Suppose Assumption~\ref{assump:1} holds. For any $\mathbf u,\mathbf w\in U(\mathbf u^*,\delta)$, let $\{\mathbf v_{\mathbf u,i}\}_{i=1}^k$ and $\{\mathbf v_{\mathbf w,i}\}_{i=1}^k$ be orthonormal eigenvectors associated with the $k$ negative eigenvalues of the respective projected Hessians. Define the orthogonal projectors
\[
\mathcal N(\mathbf u)
:=\sum_{i=1}^k
\mathbf v_{\mathbf u,i}\mathbf v_{\mathbf u,i}^{\top},
\qquad
\mathcal N(\mathbf w)
:=\sum_{i=1}^k
\mathbf v_{\mathbf w,i}\mathbf v_{\mathbf w,i}^{\top}.
\]
Then
\begin{equation}
\label{eq:eigenspace_perturbation}
\|\mathcal N(\mathbf u)-\mathcal N(\mathbf w)\|_2
\leq
\frac{2\sqrt{2k}M}{\mu}\|\mathbf u-\mathbf w\|_2.
\end{equation}
\end{corollary}

\begin{proof}
We apply the Davis--Kahan bound in~\cite[Theorem~2]{yu2015useful} by setting
\[
\bm\Sigma=\nabla_{\mathbf P}^2\mathcal E_h(\mathbf u),
\qquad
\widehat{\bm\Sigma}=\nabla_{\mathbf P}^2\mathcal E_h(\mathbf w),
\]
with eigenvector matrices
\[
\mathbf V=[\mathbf v_{\mathbf u,1},\ldots,\mathbf v_{\mathbf u,k}],
\qquad
\widehat{\mathbf V}
=[\mathbf v_{\mathbf w,1},\ldots,\mathbf v_{\mathbf w,k}].
\]
Thus, $\mathbf V\mathbf V^\top=\mathcal N(\mathbf u)$ and $\widehat{\mathbf V}\widehat{\mathbf V}^\top=\mathcal N(\mathbf w)$. Taking $r=1$ and $s=k$, the target cluster has dimension $p=k$.

By Assumption~\ref{assump:1}, the eigenvalues of $\bm\Sigma$ satisfy
\[
\lambda_1(\mathbf u)\leq\cdots\leq\lambda_k(\mathbf u)
\leq-\mu<0=\lambda_{k+1}(\mathbf u).
\]
With the convention $\lambda_0(\mathbf u)=-\infty$, the gap separating this cluster from the remaining spectrum is
\[
\begin{aligned}
\mathrm{Gap}
&:=\min\{\lambda_1(\mathbf u)-\lambda_0(\mathbf u),
          \lambda_{k+1}(\mathbf u)-\lambda_k(\mathbf u)\}\\
&=|\lambda_k(\mathbf u)|\geq\mu.
\end{aligned}
\]
Consequently,
\[
\begin{aligned}
\|\mathcal N(\mathbf u)-\mathcal N(\mathbf w)\|_2
&\leq
\|\mathcal N(\mathbf u)-\mathcal N(\mathbf w)\|_F\\
&\leq
\frac{2\sqrt{2k}}{\mathrm{Gap}}
\|\bm\Sigma-\widehat{\bm\Sigma}\|_2\\
&\leq
\frac{2\sqrt{2k}}{\mu}
\|\bm\Sigma-\widehat{\bm\Sigma}\|_2,
\end{aligned}
\]
where the first inequality uses the comparison between the spectral and Frobenius norms, the second follows from the Davis--Kahan bound, and the third uses $\mathrm{Gap}\geq\mu$. Finally, the Hessian Lipschitz estimate in Assumption~\ref{assump:1} gives
\[
\|\bm\Sigma-\widehat{\bm\Sigma}\|_2
\leq M\|\mathbf u-\mathbf w\|_2,
\]
which proves~\eqref{eq:eigenspace_perturbation}.
\end{proof}

\vspace{1em}
\subsection{Convergence Estimates for the Linearized System}

To analyze the accelerated iteration, we restrict the linearized dynamics to $\widehat{\mathcal V}_{\mathrm{eff}}$, thereby excluding the zero mode associated with mass conservation. The resulting estimate is a specialization of the heavy-ball bound in~\cite[Theorem~5]{wang2021modular}.

\begin{corollary}[Uniform heavy-ball bound]\label{cor:2}
Let $\mathbf G$ be symmetric positive definite on $\widehat{\mathcal V}_{\mathrm{eff}}$ and suppose that
\[
0<\mu\leq\lambda_{\min}(\mathbf G)\leq\lambda_{\max}(\mathbf G)\leq L,\qquad \kappa:=\frac{L}{\mu}.
\]
Define
\[
\mathbf A:=\begin{bmatrix}(1+\gamma)\mathbf I-\beta\mathbf G & -\gamma\mathbf I\\ \mathbf I & \mathbf 0\end{bmatrix},
\qquad
h(\gamma,z):=-\bigl(\gamma-(1-\sqrt z)^2\bigr)\bigl(\gamma-(1+\sqrt z)^2\bigr).
\]
If $\beta>0$, $0<\gamma\leq1$, and
\[
\gamma>\max\left\{(1-\sqrt{\beta\mu})^2,(1-\sqrt{\beta L})^2\right\},
\]
then, for every integer $m\geq0$,
\[
\|\mathbf A^m\|_2\leq C_0(\sqrt{\gamma})^m,
\qquad
C_0:=\frac{2(1+\gamma)}{\sqrt{\min\{h(\gamma,\beta\mu),h(\gamma,\beta L)\}}}.
\]
In particular, consider the parameter choice
\begin{equation}\label{eq:params}
\sqrt{\beta}=\frac{2}{\sqrt L+\sqrt\mu},\qquad \sqrt{\gamma}=1-\frac{2}{\sqrt\kappa+1}+\eta,\qquad 0<\eta<\frac{1}{\sqrt\kappa+1}.
\end{equation}
Then the above admissibility condition holds and
\[
\|\mathbf A^m\|_2\leq K(\sqrt{\gamma})^m,\qquad K:=\frac{2\sqrt{2}\,(\sqrt\kappa+1)^{1/2}}{\sqrt{3}\,\eta}.
\]
\end{corollary}

\begin{proof}
The function $\lambda\mapsto(1-\sqrt{\beta\lambda})^2$ is convex for $\lambda>0$, so the endpoint condition controls every eigenvalue of $\mathbf G$ in $[\mu,L]$. Applying the heavy-ball estimate in~\cite{wang2021modular,luo2025accelerated} and using the concavity of $h(\gamma,z)=4\gamma-(1+\gamma-z)^2$ in $z$ gives the stated bound with constant $C_0$.

For the parameter choice in~\eqref{eq:params}, set $s=\sqrt\kappa$ and $q=(s-1)/(s+1)$. Then
\[
\sqrt{\beta\mu}=\frac{2}{s+1},
\qquad
\sqrt{\beta L}=\frac{2s}{s+1},
\qquad
\sqrt\gamma=q+\eta.
\]
Since $0<\eta<1/(s+1)$, we have
\[
|1-\sqrt{\beta\mu}|
=|1-\sqrt{\beta L}|
=q<\sqrt\gamma<1,
\]
which verifies the admissibility condition.

At both endpoints, the first factor in
\[
h(\gamma,z)
=\bigl[\gamma-(1-\sqrt z)^2\bigr]
 \bigl[(1+\sqrt z)^2-\gamma\bigr]
\]
equals $\gamma-q^2=\eta(2q+\eta)\geq\eta^2$. The second factors satisfy
\[
\begin{aligned}
(1+\sqrt{\beta\mu})^2-\gamma
&=\left(\frac{4}{s+1}-\eta\right)(2+\eta)
\geq\frac{3(2+\eta)}{s+1},\\
(1+\sqrt{\beta L})^2-\gamma
&=(2-\eta)(2+2q+\eta)
\geq\frac{3(2+\eta)}{s+1},
\end{aligned}
\]
where we used $q\geq0$, $\eta<1/(s+1)$, and $2-\eta\geq3/(s+1)$ for $s\geq1$. Consequently,
\[
\min\{h(\gamma,\beta\mu),h(\gamma,\beta L)\}
\geq\frac{3\eta^2(2+\eta)}{s+1}.
\]
Substituting this estimate into $C_0$ and using $\gamma<1$ gives
\[
C_0
\leq\frac{2(1+\gamma)\sqrt{s+1}}
{\sqrt{3}\,\eta\sqrt{2+\eta}}
\leq\frac{2\sqrt{2}\sqrt{s+1}}{\sqrt{3}\,\eta}
=K,
\]
which completes the proof.
\end{proof}

\subsection{Local Convergence of the Nonlinear Iterations}

We now establish the local convergence rate of the state iteration under the exact-eigenspace assumption. To simplify the mathematical analysis and highlight the core ideas of the proof, we introduce the following standard assumption regarding the eigensolver accuracy.

\begin{assumption}[Exact Eigenvectors] \label{assump:2}
In each iteration of the algorithm, the vectors $\{\mathbf v_i^n\}_{i=1}^k$ are exact orthonormal eigenvectors of $\nabla_{\mathbf P}^2\mathcal E_h(\mathbf u^n)$ associated with its $k$ smallest eigenvalues.
\end{assumption}

The Assumption \ref{assump:2} helps to simplify the numerical analysis and highlight the techniques in the derivation of local convergence rate. In practice, the unstable eigenvectors are approximated by the direction updates in \eqref{eq:discrete_coupled_system}.
The perturbation analysis for discrete HiSD with approximated eigenvectors will lead to more technical calculations in analyzing the convergence rates, see e.g.~\cite{luo2022convergence}. This extension will be investigated in the future. 

We reformulate the discrete dynamical system into a matrix recurrence form with a remainder term.

\begin{theorem}\label{thm:3}
Suppose Assumptions~\ref{assump:1} and~\ref{assump:2} hold. Let $\mathbf u^*$ be the target constrained index-$k$ saddle point, and assume $\mathbf u^{-1}=\mathbf u^0$ and $\mathbf u^0-\mathbf u^*\in\widehat{\mathcal V}_{\mathrm{eff}}$. Define
\[
\mathbf e_n:=\mathbf u^n-\mathbf u^*,
\qquad
\mathcal R_n:=\mathbf I-2\mathcal N(\mathbf u^n),
\qquad
\mathcal R^*:=\mathbf I-2\mathcal N(\mathbf u^*),
\]
where $\mathcal N(\mathbf u)$ is the unstable eigenspace projector defined in Corollary~\ref{cor:1}. Whenever $\mathbf u^n\in U(\mathbf u^*,\delta)$, the state iteration satisfies
\begin{equation}\label{eq:dynamics_matrix}
\begin{bmatrix}
\mathbf e_{n+1}\\
\mathbf e_n
\end{bmatrix}
=
\begin{bmatrix}
(1+\gamma)\mathbf I-\beta\mathbf A_* & -\gamma\mathbf I\\
\mathbf I & \mathbf 0
\end{bmatrix}
\begin{bmatrix}
\mathbf e_n\\
\mathbf e_{n-1}
\end{bmatrix}
+
\begin{bmatrix}
\mathbf r^n\\
\mathbf 0
\end{bmatrix},
\end{equation}
where
\[
\begin{aligned}
\mathbf A_*&:=\mathcal R^*\nabla_{\mathbf P}^2\mathcal E_h(\mathbf u^*),\\
\widetilde{\mathbf A}_n&:=\mathcal R_n\int_0^1
\nabla_{\mathbf P}^2\mathcal E_h(\mathbf u^*+t\mathbf e_n)\,dt,\\
\mathbf r^n&:=\beta(\mathbf A_*-\widetilde{\mathbf A}_n)\mathbf e_n.
\end{aligned}
\]
Moreover,
\begin{equation}\label{eq:pn_bound}
\|\mathbf r^n\|_2\leq\beta C_1\|\mathbf e_n\|_2^2,
\qquad
C_1:=\left(\frac{4\sqrt{2k}L}{\mu}+\frac12\right)M.
\end{equation}
\end{theorem}

\begin{proof}
Mass conservation and the initialization imply $\mathbf e_n\in\widehat{\mathcal V}_{\mathrm{eff}}$. On this subspace, the derivative of the projected gradient is the projected Hessian. If $\|\mathbf e_n\|_2<\delta$, then $\mathbf u^\ast+t\mathbf e_n\in U(\mathbf u^\ast,\delta)$ for every $t\in[0,1]$. Therefore, using $\nabla_{\mathbf P}\mathcal E_h(\mathbf u^\ast)=\mathbf 0$, the fundamental theorem of calculus gives
\[
\nabla_{\mathbf P}\mathcal E_h(\mathbf u^n)
=
\left[\int_0^1
\nabla_{\mathbf P}^2\mathcal E_h(\mathbf u^*+t\mathbf e_n)\,dt
\right]\mathbf e_n.
\]
Under Assumption~\ref{assump:2}, the state update therefore becomes
\[
\begin{aligned}
\mathbf e_{n+1}
&=(1+\gamma)\mathbf e_n-\gamma\mathbf e_{n-1}
-\beta\mathcal R_n\nabla_{\mathbf P}\mathcal E_h(\mathbf u^n)\\
&=\bigl((1+\gamma)\mathbf I-\beta\mathbf A_*\bigr)\mathbf e_n
-\gamma\mathbf e_{n-1}+\mathbf r^n,
\end{aligned}
\]
which yields~\eqref{eq:dynamics_matrix}.

To estimate the remainder, write
\[
\begin{aligned}
\mathbf A_*-\widetilde{\mathbf A}_n
={}&(\mathcal R^*-\mathcal R_n)
\nabla_{\mathbf P}^2\mathcal E_h(\mathbf u^*)\\
&+\mathcal R_n\int_0^1
\left[
\nabla_{\mathbf P}^2\mathcal E_h(\mathbf u^*)
-\nabla_{\mathbf P}^2\mathcal E_h(\mathbf u^*+t\mathbf e_n)
\right]dt.
\end{aligned}
\]
The reflection satisfies $\|\mathcal R_n\|_2=1$, while Corollary~\ref{cor:1} gives
\[
\|\mathcal R^*-\mathcal R_n\|_2
=2\|\mathcal N(\mathbf u^*)-\mathcal N(\mathbf u^n)\|_2
\leq\frac{4\sqrt{2k}M}{\mu}\|\mathbf e_n\|_2.
\]
Using the spectral bound and Hessian Lipschitz continuity in Assumption~\ref{assump:1}, we obtain
\[
\begin{aligned}
\|\mathbf A_*-\widetilde{\mathbf A}_n\|_2
&\leq L\|\mathcal R^*-\mathcal R_n\|_2
+\int_0^1 Mt\|\mathbf e_n\|_2\,dt\\
&\leq\left(\frac{4\sqrt{2k}LM}{\mu}+\frac M2\right)
\|\mathbf e_n\|_2
=C_1\|\mathbf e_n\|_2.
\end{aligned}
\]
Consequently,
\[
\|\mathbf r^n\|_2
\leq\beta\|\mathbf A_*-\widetilde{\mathbf A}_n\|_2
\|\mathbf e_n\|_2
\leq\beta C_1\|\mathbf e_n\|_2^2,
\]
proving~\eqref{eq:pn_bound}.
\end{proof}

\vspace{1em}
For the convergence analysis, define
\begin{equation}\label{eq:21}
\mathbf Z_{n+1}
:=
\begin{bmatrix}
\mathbf u^{n+1}-\mathbf u^\ast\\
\mathbf u^{n}-\mathbf u^\ast
\end{bmatrix},
\qquad
\mathbf T
:=
\begin{bmatrix}
(1+\gamma)\mathbf I-\beta\mathbf A_\ast & -\gamma\mathbf I\\
\mathbf I & \mathbf 0
\end{bmatrix},
\qquad
\mathbf R_n
:=
\begin{bmatrix}
\mathbf r^n\\
\mathbf 0
\end{bmatrix}.
\end{equation}
Then \eqref{eq:dynamics_matrix} can be written as
\begin{equation}\label{eq:block_recurrence}
\mathbf Z_{n+1}=\mathbf T\mathbf Z_n+\mathbf R_n.
\end{equation}

\begin{lemma}[Effective Subspace Isometry] \label{lemma:isometry}
Let $\mathbf{V} \in \mathbb{R}^{N_{in} \times (N_{in}-1)}$ be an orthonormal basis matrix for the effective subspace $\hat{\mathcal{V}}_{\text{eff}}$, satisfying $\mathbf{V}^\top \mathbf{V} = \mathbf{I}_{N_{in}-1}$. Define the restricted limiting matrix as $\mathbf{\tilde{A}}_* := \mathbf{V}^\top \mathbf{A}_* \mathbf{V} \in \mathbb{R}^{(N_{in}-1) \times (N_{in}-1)}$, and the corresponding restricted propagation matrix as:
\[
\mathbf{\tilde{T}} := \begin{bmatrix} (1 + \gamma)\mathbf{I}_{N_{in}-1} - \beta \mathbf{\tilde{A}}_* & -\gamma \mathbf{I}_{N_{in}-1} \\ \mathbf{I}_{N_{in}-1} & \mathbf{0} \end{bmatrix}.
\]
Define the augmented orthogonal basis matrix $\mathbf{\mathbb{V}} := \begin{bmatrix} \mathbf{V} & \mathbf{0} \\ \mathbf{0} & \mathbf{V} \end{bmatrix}$. If an initial state vector $\mathbf{Z}_0$ lies entirely within the effective subspace $\hat{\mathcal{V}}_{\text{eff}} \times \hat{\mathcal{V}}_{\text{eff}}$, meaning that there exists a unique coordinate vector $\tilde{\mathbf{Z}}_0 \in \mathbb{R}^{2(N_{in}-1)}$ such that $\mathbf{Z}_0 = \mathbb{V} \tilde{\mathbf{Z}}_0$, then the full-space dynamics and the restricted dynamics satisfy the intertwining relation $\mathbf{T}^m \mathbf{\mathbb{V}} = \mathbf{\mathbb{V}} \mathbf{\tilde{T}}^m$ for any $m \ge 0$. Consequently, the following isometric identity holds strictly:
\begin{equation} \label{eq:norm_equiv}
\|\mathbf{T}^{n+1} \mathbf{Z}_0\|_2 = \|\mathbf{\mathbb{V}} \mathbf{\tilde{T}}^{n+1} \mathbf{\tilde{Z}}_0\|_2 = \|\mathbf{\tilde{T}}^{n+1} \mathbf{\tilde{Z}}_0\|_2.
\end{equation}
\end{lemma}

\begin{proof}
Because $\hat{\mathcal{V}}_{\text{eff}}$ is an invariant subspace of $\mathbf{A}_*$, we have $\mathbf{A}_* \mathbf{V} = \mathbf{V} \mathbf{\tilde{A}}_*$. Therefore, we can establish the intertwining relation between $\mathbf{T}$ and $\mathbf{\tilde{T}}$ as follows:
\begin{align*}
\mathbf{T} \mathbf{\mathbb{V}} &= \begin{bmatrix} (1 + \gamma)\mathbf{I} - \beta \mathbf{A}_* & -\gamma \mathbf{I} \\ \mathbf{I} & \mathbf{0} \end{bmatrix} \begin{bmatrix} \mathbf{V} & \mathbf{0} \\ \mathbf{0} & \mathbf{V} \end{bmatrix} \\
&= \begin{bmatrix} ((1 + \gamma)\mathbf{I} - \beta \mathbf{A}_*)\mathbf{V} & -\gamma \mathbf{V} \\ \mathbf{V} & \mathbf{0} \end{bmatrix} \\
&= \begin{bmatrix} \mathbf{V}((1 + \gamma)\mathbf{I}_{N_{in}-1} - \beta \mathbf{\tilde{A}}_*) & -\gamma \mathbf{V} \\ \mathbf{V} & \mathbf{0} \end{bmatrix} \\
&= \begin{bmatrix} \mathbf{V} & \mathbf{0} \\ \mathbf{0} & \mathbf{V} \end{bmatrix} \begin{bmatrix} (1 + \gamma)\mathbf{I}_{N_{in}-1} - \beta \mathbf{\tilde{A}}_* & -\gamma \mathbf{I}_{N_{in}-1} \\ \mathbf{I}_{N_{in}-1} & \mathbf{0} \end{bmatrix} \\
&= \mathbf{\mathbb{V}} \mathbf{\tilde{T}}.
\end{align*}
Through simple mathematical induction, this relation extends to any power: $\mathbf{T}^{n+1} \mathbf{\mathbb{V}} = \mathbf{\mathbb{V}} \mathbf{\tilde{T}}^{n+1}$. Consequently, the full-space error evolution can be expressed as:
$$\mathbf{T}^{n+1} \mathbf{Z}_0 = \mathbf{T}^{n+1} \mathbf{\mathbb{V}} \mathbf{\tilde{Z}}_0 = \mathbf{\mathbb{V}} \mathbf{\tilde{T}}^{n+1} \mathbf{\tilde{Z}}_0.$$
Since $\mathbf{V}^\top \mathbf{V} = \mathbf{I}_{N_{in}-1}$, the augmented matrix $\mathbf{\mathbb{V}}$ is also an isometry satisfying $\mathbf{\mathbb{V}}^\top \mathbf{\mathbb{V}} = \mathbf{I}_{2(N_{in}-1)}$. Because an isometric map strictly preserves the 2-norm of a vector (i.e., $\|\mathbf{\mathbb{V}}\mathbf{y}\|_2 = \|\mathbf{y}\|_2$), we obtain the strict norm equivalence \eqref{eq:norm_equiv}.
\end{proof}

\begin{theorem}[Local convergence rate of DD-ASD under the exact-eigenspace assumption]\label{thm:4}
Suppose Assumptions~\ref{assump:1} and~\ref{assump:2} hold, and assume that the target saddle point $\mathbf{u}^*$ has the same prescribed discrete mass as $\mathbf{u}^0$, or equivalently, $\mathbf{u}^0-\mathbf{u}^*\in\widehat{\mathcal V}_{\mathrm{eff}}$. Let $\eta\in\left(0,1/(\sqrt{\kappa}+1)\right)$, set $\mathbf{u}^{-1}=\mathbf{u}^0$, and choose $\beta$ and $\gamma$ according to~\eqref{eq:params}. If
\[
\|\mathbf{u}^0-\mathbf{u}^*\|_2\leq\frac{\widehat r}{2\sqrt{2}K},
\qquad
\widehat r<\min\left\{\delta,\frac{\sqrt{3}\mu\eta^2(\sqrt{\kappa}+1)^{3/2}}{16\sqrt{2}C_1}\right\},
\]
then the DD-ASD iterates converge to $\mathbf{u}^*$ and satisfy
\begin{equation}\label{eq:22}
\|\mathbf{u}^n-\mathbf{u}^*\|_2\leq2\sqrt{2}K\|\mathbf{u}^0-\mathbf{u}^*\|_2\rho^n,\qquad \rho:=1-\frac{2}{\sqrt{\kappa}+1}+2\eta<1.
\end{equation}
Here, $C_1$ is defined in Theorem~\ref{thm:3}, and $K$ is given in Corollary~\ref{cor:2}.
\end{theorem}

\begin{proof}
By Theorem~\ref{thm:3}, the error satisfies
\[
\mathbf Z_{n+1}=\mathbf T\mathbf Z_n+\mathbf R_n,\qquad \|\mathbf R_n\|_2=\|\mathbf r^n\|_2\leq\beta C_1\|\mathbf u^n-\mathbf u^*\|_2^2
\]
whenever $\|\mathbf u^n-\mathbf u^*\|_2<\delta$. Iterating the recurrence gives
\begin{equation}\label{eq:23}
\mathbf Z_{n+1}=\mathbf T^{n+1}\mathbf Z_0+\sum_{j=0}^n\mathbf T^{n-j}\mathbf R_j.
\end{equation}

The same-mass condition and the invariance of the projected iteration imply that $\mathbf u^n-\mathbf u^*$ and $\mathbf r^n$ lie in $\widehat{\mathcal V}_{\mathrm{eff}}$. Let $\mathbf V$ be the orthonormal basis matrix introduced in Lemma~\ref{lemma:isometry}. The restriction of $\mathbf A_*=\mathcal R^*\nabla_{\mathbf P}^2\mathcal E_h(\mathbf u^*)$ to this subspace is
\[
\widetilde{\mathbf A}_*:=\mathbf V^\top\mathbf A_*\mathbf V,
\qquad
\sigma(\widetilde{\mathbf A}_*)=\{-\lambda_1^*,\ldots,-\lambda_k^*,\lambda_{k+2}^*,\ldots,\lambda_{N_{in}}^*\}\subset[\mu,L].
\]
Thus, $\widetilde{\mathbf A}_*$ is positive definite. Lemma~\ref{lemma:isometry} and Corollary~\ref{cor:2} therefore give, for vectors in $\widehat{\mathcal V}_{\mathrm{eff}}\times\widehat{\mathcal V}_{\mathrm{eff}}$,
\[
\|\mathbf T^m\mathbf Y\|_2\leq K(\sqrt{\gamma})^m\|\mathbf Y\|_2,\qquad m\geq0.
\]
Applying this estimate to~\eqref{eq:23} yields
\begin{equation}\label{eq:propagation_bound}
\|\mathbf Z_{n+1}\|_2\leq K(\sqrt{\gamma})^{n+1}\|\mathbf Z_0\|_2+K\sum_{j=0}^n(\sqrt{\gamma})^{n-j}\|\mathbf r^j\|_2.
\end{equation}

We prove simultaneously that
\[
\|\mathbf Z_n\|_2\leq2K\rho^n\|\mathbf Z_0\|_2
\qquad\text{and}\qquad
\|\mathbf Z_n\|_2\leq\widehat r
\]
for all $n\geq0$. Since $\mathbf u^{-1}=\mathbf u^0$,
\[
\|\mathbf Z_0\|_2=\sqrt{2}\|\mathbf u^0-\mathbf u^*\|_2\leq\frac{\widehat r}{2K},
\]
so both estimates hold for $n=0$. Suppose they hold for $0\leq j\leq m$. Then $\|\mathbf u^j-\mathbf u^*\|_2\leq\widehat r<\delta$, and~\eqref{eq:propagation_bound} gives
\begin{align}
\|\mathbf Z_{m+1}\|_2
&\leq K(\sqrt{\gamma})^{m+1}\|\mathbf Z_0\|_2+K\beta C_1\sum_{j=0}^m(\sqrt{\gamma})^{m-j}\|\mathbf Z_j\|_2^2 \nonumber\\
&\leq K\rho^{m+1}\|\mathbf Z_0\|_2+2\beta K^2C_1\widehat r\sum_{j=0}^m(\sqrt{\gamma})^{m-j}\rho^j\|\mathbf Z_0\|_2 \nonumber\\
&\leq\left(1+\frac{2\beta KC_1\widehat r}{\rho-\sqrt{\gamma}}\right)K\rho^{m+1}\|\mathbf Z_0\|_2. \label{eq:27}
\end{align}
Since $\rho-\sqrt{\gamma}=\eta$, the definition of $\widehat r$ implies
\[
\frac{2\beta KC_1\widehat r}{\rho-\sqrt{\gamma}}
=
\frac{16\sqrt{2}C_1\widehat r}{\sqrt{3}\mu\eta^2(\sqrt{\kappa}+1)^{3/2}}
<1.
\]
Consequently,
\[
\|\mathbf Z_{m+1}\|_2\leq2K\rho^{m+1}\|\mathbf Z_0\|_2\leq\widehat r,
\]
where the second inequality follows from $\rho<1$ and $\|\mathbf Z_0\|_2\leq\widehat r/(2K)$. The two estimates therefore hold for all $n\geq0$ by induction. Finally,
\[
\|\mathbf u^n-\mathbf u^*\|_2\leq\|\mathbf Z_n\|_2\leq2K\rho^n\|\mathbf Z_0\|_2=2\sqrt{2}K\rho^n\|\mathbf u^0-\mathbf u^*\|_2,
\]
which proves~\eqref{eq:22}.
\end{proof}

\begin{remark}
In Theorem \ref{thm:4}, if we select $\eta = \frac{1}{2(\sqrt{\kappa}+1)}$, the local linear convergence rate becomes $\rho = 1 - \frac{1}{\sqrt{\kappa}+1}$. Compared to the proven convergence rate of $1 - \mathcal{O}(\frac{1}{\kappa})$ for the standard discrete HiSD algorithm, the DD-ASD algorithm equipped with momentum acceleration achieves an order of $1 - \mathcal{O}(\frac{1}{\sqrt{\kappa}})$. This yields substantial acceleration, especially under large condition numbers $\kappa$ (ill-conditioned systems) typical in diffuse-domain phase-field models.
\end{remark}
\section{Numerical Experiments}


This section examines the accuracy, structure preservation, efficiency, and applicability of the proposed DD-ASD framework for wetting transition problems.
Unless otherwise stated, the phase-field width is fixed at $\epsilon=0.015$, and the spatial grids are chosen to resolve the diffuse interface.
\subsection{Comprehensive Validation of the DD-ASD Framework}

To validate the proposed DD-ASD framework, we evaluate its performance from both physical and numerical perspectives. First, we verify the algorithm's capability to accurately resolve complex physical solution landscapes and their underlying transition pathways. Subsequently, we examine the step-size self-convergence of the discrete
iteration and verify numerically that mass is conserved.
\subsubsection{Solution Landscape as a Physical Benchmark}
\begin{figure}[H]
  \centering
  \includegraphics[page=1, trim=0cm 14cm 1cm 0.5cm, clip, width=0.6\textwidth]{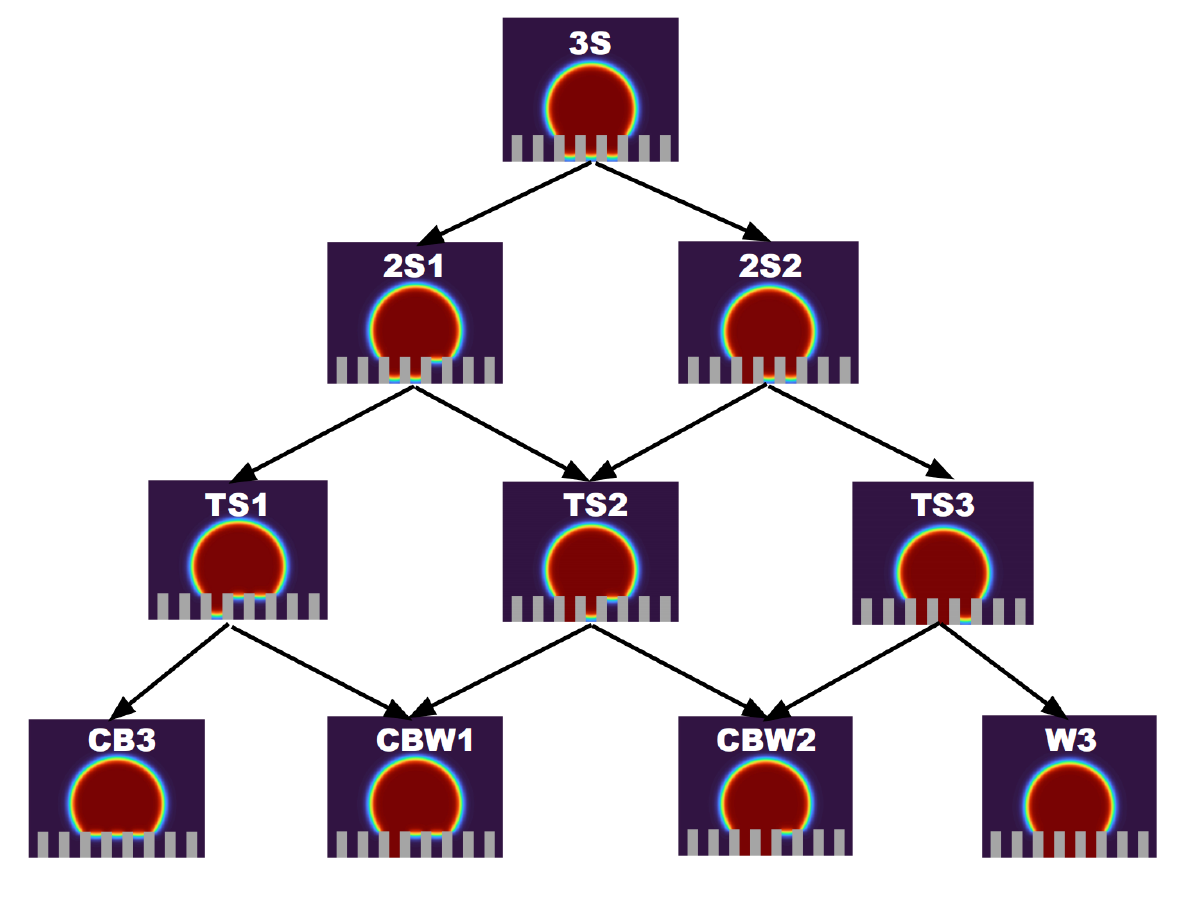}
  \caption{Solution landscape for the wetting transition obtained with the pillar gap $w=0.06$, pillar depth $d=0.15$, and Young's angle $\theta=102^\circ$.}
  \label{fig:solution_landscape}
\end{figure}
We first deploy our solver to  compute the solution landscape of a droplet on a micro-structured substrate. To ensure a direct and meaningful comparison, the geometric configurations (pillar gap $w=0.06$, depth $d=0.15$) and the physical parameter (Young's angle $\theta=102^\circ$) are chosen to be strictly identical to those investigated via conventional FEM in recent literature \cite{zhang2025solution}. 

It must be emphasized that resolving this highly non-convex, multi-stable wetting network serves as a stringent physical benchmark. The algorithm must precisely capture not only the stable Cassie-Baxter and Wenzel minima, but also the intermediate transition pathways and saddle points of index up to three. As illustrated in Figure~\ref{fig:solution_landscape}, the multi-stable solution landscape is successfully resolved by our DD-ASD framework. The calculated pathways and topological connections exhibit good agreement with the reference FEM data \cite{zhang2025solution} and theoretical expectations \cite{bormashenko2015progress}, supporting the ability of DD-ASD to recover the reported wetting transition landscape.

\subsubsection{Discrete Step Size Self-Convergence}

We assess the step size self-convergence of the discrete DD-ASD iteration. In this experiment, $h=0.005$ and $\gamma=0.5$ are fixed, and we write $\beta=\Delta t$. A reference solution is computed with $\Delta t_{\mathrm{ref}}=5\times10^{-6}$, and all computations are advanced to $T=5$.

Table~\ref{tab:time_convergence} presents the $L^2$ and $L^\infty$ errors and the corresponding convergence orders for $h=0.005$ and $\gamma=0.5$. The results exhibit approximately first-order self-convergence of the discrete iteration over the tested range of step sizes.
\begin{table}[htbp]
  \centering
  \caption{Errors relative to the reference solution and observed step-size self-convergence orders.}
  \label{tab:time_convergence}
  \renewcommand{\arraystretch}{1.15} 
  \begin{tabular}{@{} c c c c c @{}}
    \toprule
    \textbf{$\Delta t$} & \textbf{$L^2$ Error} & \textbf{Order} & \textbf{$L^\infty$ Error} & \textbf{Order} \\
    \midrule
    $2\times10^{-4}$   & $1.2223 \times 10^{-6}$ & ---   & $1.3481 \times 10^{-5}$ & ---   \\
    $1\times10^{-4}$   & $5.8097 \times 10^{-7}$ & 1.0730 & $6.4078 \times 10^{-6}$ & 1.0730 \\
    $5\times10^{-5}$   & $2.7180 \times 10^{-7}$ & 1.0959 & $2.9978 \times 10^{-6}$ & 1.0959 \\
    $2.5\times10^{-5}$ & $1.2003 \times 10^{-7}$ & 1.1791 & $1.3239 \times 10^{-6}$ & 1.1791 \\
    \bottomrule
  \end{tabular}
\end{table}


\subsubsection{Strict Mass Conservation}

Mass conservation is a fundamental requirement in modeling droplet transitions. To examine this property numerically, we record the discrete mass every 200 iterations during each saddle-point search.

Figure~\ref{fig:mass_error} reports the deviation from the initial mass for the searches of the index-1, index-2, and index-3 saddle points. At each recorded iteration $n$, the mass deviation is defined by $\Delta M_h^n=M_h(\mathbf u^n)-M_h(\mathbf u^0)$. Each search is continued until the projected gradient norm reaches the prescribed tolerance of $10^{-6}$. The deviations remain of order $10^{-15}$ throughout the computations, consistent with floating-point roundoff errors and the exact-arithmetic conservation result in Theorem~\ref{thm:mass_conservation}.

\begin{figure}[h]
  \centering
  \begin{minipage}[b]{0.32\textwidth}
    \centering
    \includegraphics[width=\textwidth]{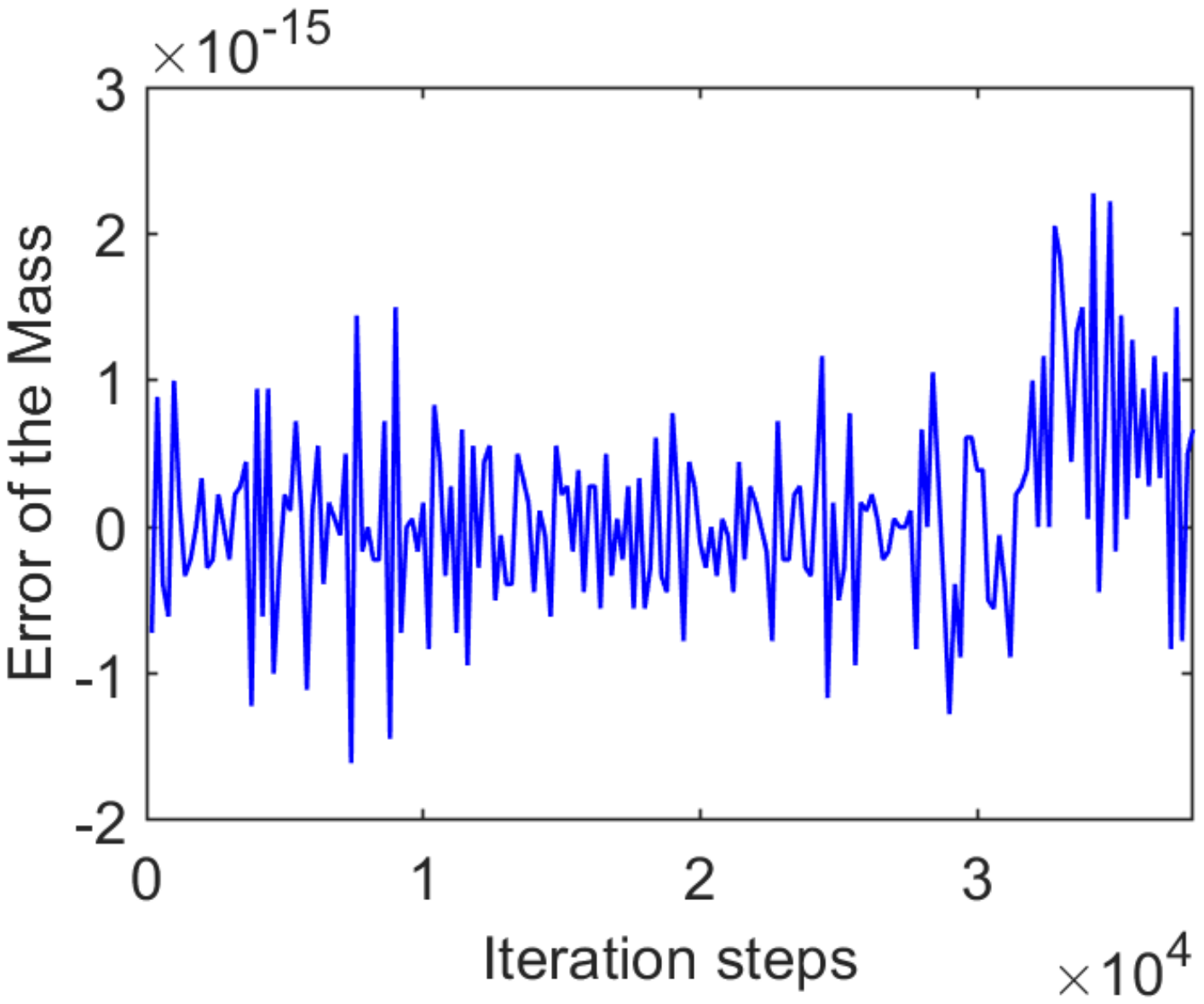} 
    \vspace{0.1cm}
    \centerline{(a) Index-1 saddle}
  \end{minipage}
  \hfill
  \begin{minipage}[b]{0.32\textwidth}
    \centering
    \includegraphics[width=\textwidth]{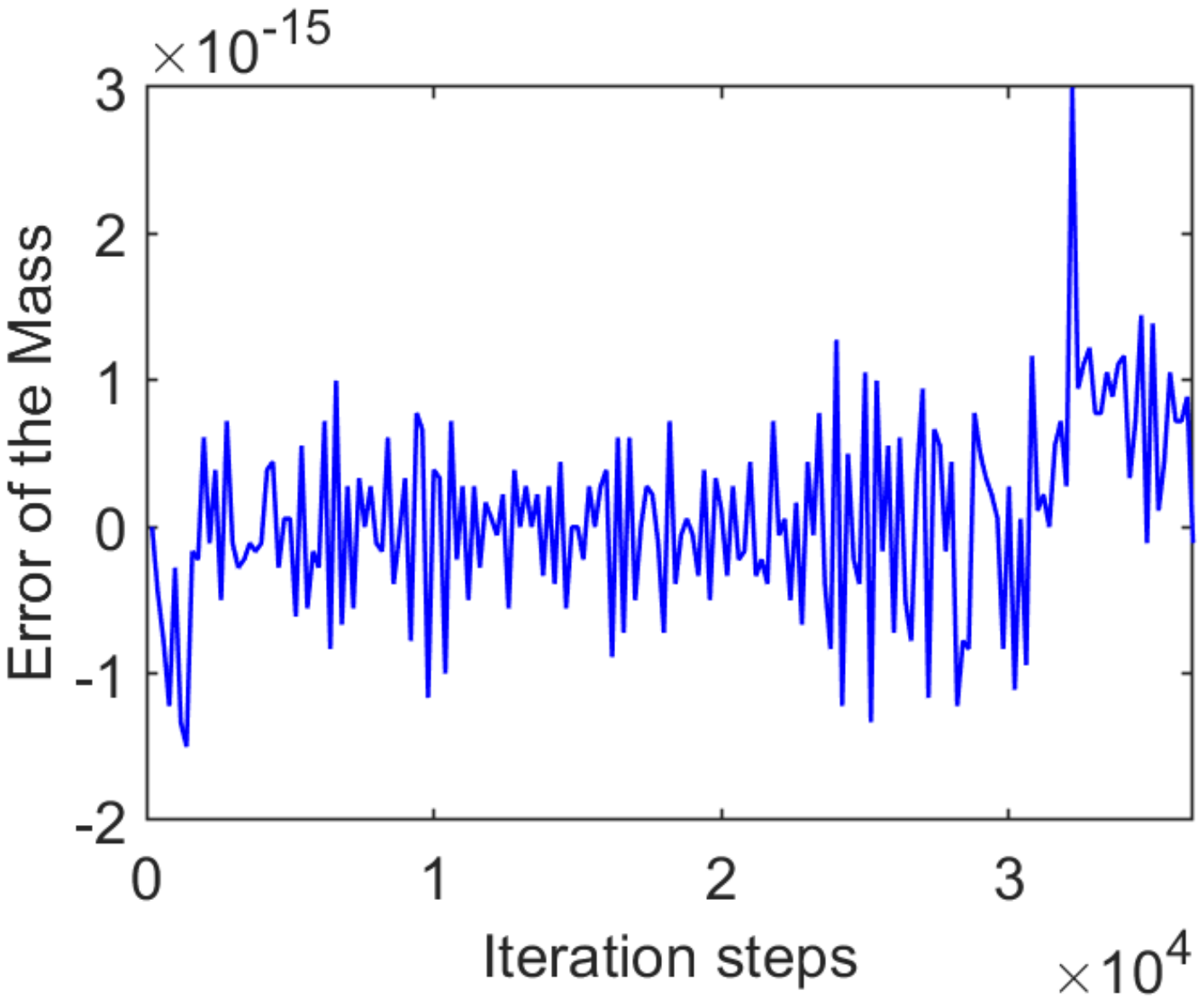} 
    \vspace{0.1cm}
    \centerline{(b) Index-2 saddle}
  \end{minipage}
  \hfill
  \begin{minipage}[b]{0.32\textwidth}
    \centering
    \includegraphics[width=\textwidth]{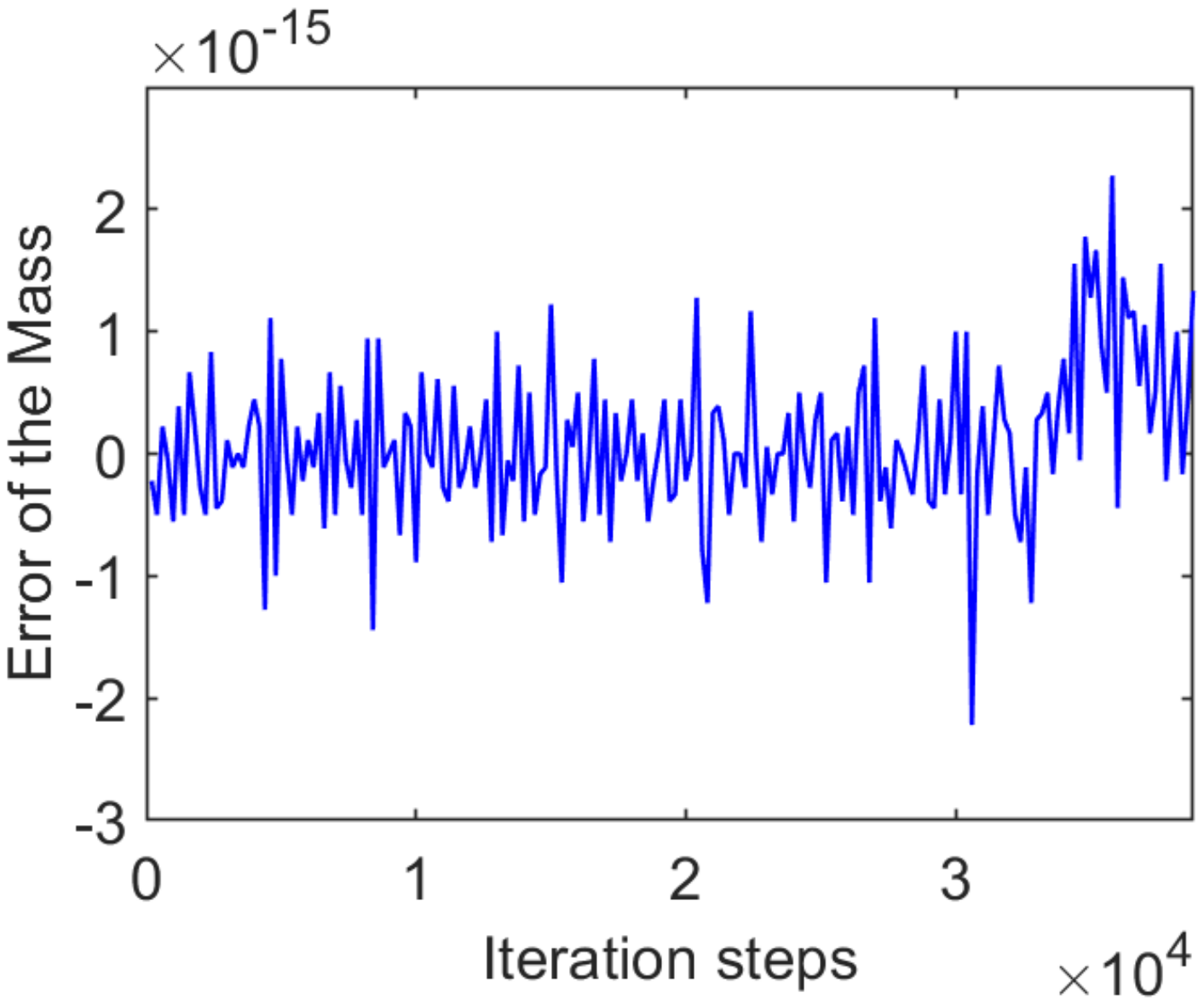} 
    \vspace{0.1cm}
    \centerline{(c) Index-3 saddle}
  \end{minipage}
  \caption{Evolution of the mass variation for the index-1, index-2, and
index-3 saddle-point searches, evaluated every 200 iterations.}
  \label{fig:mass_error}
\end{figure}

\subsubsection{Mesh Sensitivity for a Representative W3--TS3 Transition}
As a representative test of mesh sensitivity, we repeat one W3--TS3 transition identified in the preceding solution-landscape experiment. We compare the resulting TS3 configuration and the corresponding energy barrier on four Cartesian grids with $h=0.015$, $0.010$, $0.0075$, and $0.005$. The phase-field width, Young's angle, and substrate geometry are fixed at the values used in Figure \ref{fig:solution_landscape}. On each grid, the computation starts from the corresponding W3 state,
and both W3 and TS3 satisfy the discrete mass constraint $M_h(\mathbf u_h)=h^2\sum_{i,j}(u_h)_{i,j}(\psi_b)_{i,j}=M_0$. Each search is terminated when $\|\nabla_{\mathbf P}\mathcal E_h(\mathbf u_h)\|_2<10^{-6}$, where $\|\cdot\|_2$ denotes the Euclidean norm.

Starting from the W3 configuration on each grid, the upward searches recover the characteristic TS3 state, as shown in Figure~\ref{fig:compare}. The resulting configurations exhibit the same qualitative pattern of liquid across the tested resolutions.

\begin{figure}[htbp]
    \centering
    \includegraphics[
        page=1,
        trim=1cm 26cm 2cm 0cm,
        clip,
        width=0.95\textwidth
    ]{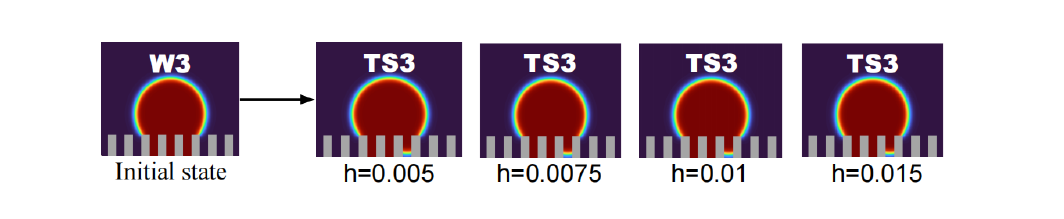}
    \caption{
    Comparison of the TS3 configurations computed on different
    Cartesian grids. The leftmost panel shows W3, from which
    the upward searches are initiated. The remaining panels
    show the resulting TS3 configurations for
    $h=0.005$, $0.0075$, $0.010$, and $0.015$,
    from left to right.
    }
    \label{fig:compare}
\end{figure}

Table~\ref{tab:energy_barrier_mesh} reports the computed barriers $
\Delta\mathcal{E}_h
=
\mathcal{E}_h(u_h^{\mathrm{TS3}})
-
\mathcal{E}_h(u_h^{\mathrm{W3}})
$ and their relative differences from the finest-grid result, with $D_h = |\Delta\mathcal{E}_h - \Delta\mathcal{E}_{h_{\mathrm{ref}}}| / |\Delta\mathcal{E}_{h_{\mathrm{ref}}}|$.

\begin{table}[htbp]
    \centering
    \caption{
    W3--TS3 energy barriers under spatial refinement.
    Relative differences are measured against the result
    at $h_{\mathrm{ref}}=0.005$.
    }
    \label{tab:energy_barrier_mesh}
    \renewcommand{\arraystretch}{1.15}
    \begin{tabular}{@{}ccc@{}}
        \toprule
        Mesh size $h$
        & Energy barrier $\Delta\mathcal{E}_h$
        & Relative difference $D_h$ \\
        \midrule
        0.0150 & 0.0437 & 7.11\% \\
        0.0100 & 0.0423 & 3.68\% \\
        0.0075 & 0.0417 & 2.21\% \\
        0.0050 & 0.0408 & ---    \\
        \bottomrule
    \end{tabular}
\end{table}

The barrier decreases monotonically from $0.0437$ to $0.0408$ as the mesh is refined. The two finest grids differ by $9.0\times10^{-4}$, corresponding to approximately $2.21\%$ of the finest-grid barrier. Together, the configuration comparison and the barrier values show that the qualitative TS3 morphology is preserved across the tested grids, while the barrier retains a measurable dependence on spatial resolution.

\subsection{Computational Efficiency and Algorithmic Acceleration}

We compare FEM-SD with DD-SD to assess the effect of the Cartesian diffuse-domain discretization, and DD-SD with DD-ASD to assess the additional effect of momentum.

All computations use
$\Delta t=3\times10^{-4}$ and stop when
$\|\nabla_{\mathbf P}\mathcal E_h\|_2<10^{-6}$.
The reported times include only the main iteration and exclude
initialization and post-processing. The comparisons follow the upward
pathway
$\mathrm{W3}\rightarrow\mathrm{TS3}\rightarrow
\mathrm{2S2}\rightarrow\mathrm{3S}$
shown in Figure~\ref{fig:solution_landscape}.

\begin{table}[htbp]
\centering
\caption{Computation times (in seconds) for different methods and
spatial resolutions.}
\label{tab:computational_times}
\renewcommand{\arraystretch}{1.15}
\setlength{\tabcolsep}{4.5pt}
\small
\begin{tabular}{@{} l c r r r @{}}
\toprule
\textbf{Method}
& \textbf{Mesh Size}
& \shortstack{\textbf{Index-1}\\\textbf{(TS3)}}
& \shortstack{\textbf{Index-2}\\\textbf{(2S2)}}
& \shortstack{\textbf{Index-3}\\\textbf{(3S)}} \\
\midrule
FEM-SD
& $h_{\max}=0.015$
& 1974.565 & 3583.747 & 3435.885 \\
\midrule
DD-SD
& $h=0.015$
& 32.897 & 77.899 & 250.428 \\
DD-SD
& $h=0.005$
& 359.447 & 835.192 & 1066.751 \\
\midrule
DD-ASD ($\gamma=0.4$)
& $h=0.015$
& 17.013 & 35.488 & 133.749 \\
DD-ASD ($\gamma=0.4$)
& $h=0.005$
& 241.473 & 511.706 & 669.214 \\
\bottomrule
\end{tabular}
\end{table}

Table~\ref{tab:computational_times} shows that DD-SD requires less
computation time than FEM-SD in all three searches. At the nominal mesh
scale $0.015$, DD-SD takes $32.897$--$250.428$\,s, compared with
$1974.565$--$3583.747$\,s for FEM-SD. Even with the finer Cartesian grid
$h=0.005$, the DD-SD computations remain faster than the FEM-SD
benchmark. This comparison shows that part of the reduction in
computational cost comes from the diffuse-domain discretization itself.

Momentum gives a further reduction in computation time, but its effect
depends on $\gamma$. For the representative index-1 search at $h=0.015$,
the computation times for $\gamma=0.1$, $0.2$, $0.3$, and $0.4$ are
$26.426$, $23.326$, $21.556$, and $17.013$\,s, respectively. Further increasing $\gamma$ did not reduce the computation time in our tests and
could lead to overshooting during the early stage of the upward search.
We therefore use $\gamma=0.4$ in Table~\ref{tab:computational_times}.
Compared with DD-SD, DD-ASD reduces the computation times by
$48.3\%$, $54.4\%$, and $46.6\%$ at $h=0.015$, and by $32.8\%$,
$38.7\%$, and $37.3\%$ at $h=0.005$, for the index-1, index-2, and
index-3 searches, respectively.

\begin{remark}[Adaptive momentum strategy]
\label{rmk:adaptive_momentum}
The suitable value of $\gamma$ may change during the saddle search.
A smaller value is used in the initial stage to reduce possible
overshooting. After the projected gradient norm passes its initial peak
and starts to decrease, a larger value can be used to speed up the later
iterations. We use the following switching rule.

Let
\[
a_n:=\|\nabla_{\mathbf P}\mathcal E_h(\mathbf u^n)\|_2,
\qquad
N_{\mathrm{peak}}
:=
\min\left\{
n\ge2:
a_n<a_{n-1},\;
a_{n-1}>a_{n-2}
\right\}.
\]
Define
\begin{equation}
N_{\mathrm{switch}}
:=
\min\left\{
n\ge N_{\mathrm{peak}}:
a_n<1,\;
a_n<a_{n-1}
\right\},
\end{equation}
and set
\begin{equation}
\gamma^n=
\begin{cases}
\gamma_{\mathrm{init}}, & n<N_{\mathrm{switch}},\\
\gamma_{\mathrm{tar}}, & n\ge N_{\mathrm{switch}}.
\end{cases}
\end{equation}
If no switching index is found, $\gamma^n$ remains equal to
$\gamma_{\mathrm{init}}$. This switching rule is used only in the
numerical implementation. The convergence analysis in Section~3 is for
a fixed momentum parameter.
\end{remark}

\subsection{Applications to Complex Pillar Designs and Parametric Analysis}
\subsubsection{Adaptability to Diverse Pillar Geometries}

The results above show that DD-ASD remains computationally feasible on the fine grid with $h=0.005$; see Table~\ref{tab:computational_times}. This allows more detailed substrate geometries to be represented on a Cartesian grid by reducing $h$, without constructing a new body-fitted unstructured mesh for each geometry.

We illustrate this geometric flexibility by computing transition pathways for substrates with rectangular, semicircular, and triangular pillars. Figure~\ref{fig:complex_geometries} shows the computed transition states, and Table~\ref{tab:energy_barriers_geometry} reports the corresponding energy barriers. Among the three geometries, the triangular pillars (Structure C) have the lowest maximum barrier, $0.0064$, whereas the rectangular pillars (Structure A) have the highest, $0.0146$. Together with the timing results, these examples show that DD-ASD can treat varied pillar geometries at a fine spatial resolution without geometry-specific remeshing.

\begin{figure}[htbp]
  \centering
  \includegraphics[width=1\textwidth]{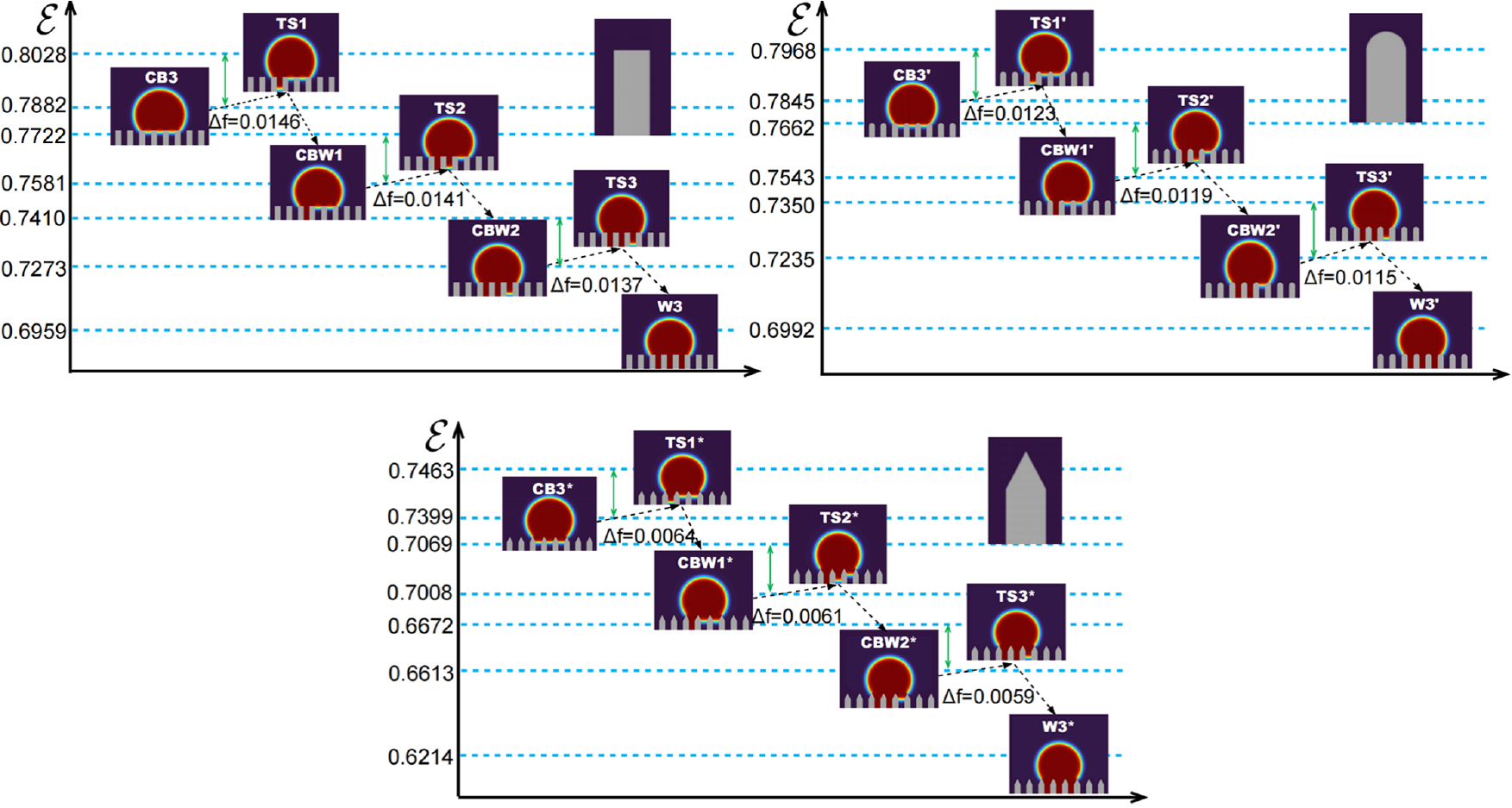}
  \caption{Transition paths of the droplet on substrates with different pillar geometries.}
  \label{fig:complex_geometries}
\end{figure}

\begin{table}[htbp]
  \centering
   \caption{Energy barriers for the three pillar geometries. Barriers 1, 2, and 3 correspond to CB3 $\to$ TS1, CBW1 $\to$ TS2, and CBW2 $\to$ TS3, respectively.}
  \label{tab:energy_barriers_geometry}
  \renewcommand{\arraystretch}{1.2} 
  \begin{tabular}{@{} l c c c @{}}
    \toprule
    \multirow{2}{*}{\textbf{Transition Phase}} & \textbf{Structure A} & \textbf{Structure B} & \textbf{Structure C} \\
    & \textbf{(Rectangular)} & \textbf{(Semicircular)} & \textbf{(Triangular)} \\
    \midrule
    Energy barrier 1   & 0.0146 & 0.0123 & 0.0064 \\
    Energy barrier 2   & 0.0141 & 0.0119 & 0.0061 \\
    Energy barrier 3   & 0.0137 & 0.0115 & 0.0059 \\
    \midrule
    \textbf{Max Energy barrier}        & \textbf{0.0146} & \textbf{0.0123} & \textbf{0.0064} \\
    \bottomrule
  \end{tabular}
\end{table}

\subsubsection{Parametric Study of Transition Energy Barriers}
We next examine how variations in geometric and physical parameters affect the wetting-transition energy barriers. Figure~\ref{fig:parameter_effects}(a) shows the dependence of the energy barriers on four parameters: Young's contact angle $\theta$, pillar gap $w$, pillar depth $d$, and central angle $\phi$. Figure~\ref{fig:parameter_effects}(b) presents the corresponding geometric configuration, illustrating how $\phi$ characterizes the curvature of the pillar caps, along with the definitions of the spacing $w$ and total depth $d$. Over the parameter ranges considered, the largest energy barrier consistently occurs during the initial transition from the CB3 state to the CBW1 state. These results clarify how the geometric and wetting parameters influence the transition barriers and provide numerical guidance for designing textured surfaces to facilitate or inhibit wetting transitions.

\begin{figure}[htbp]
  \centering
  \begin{minipage}[b]{0.7\textwidth}
    \centering
    \includegraphics[width=\textwidth]{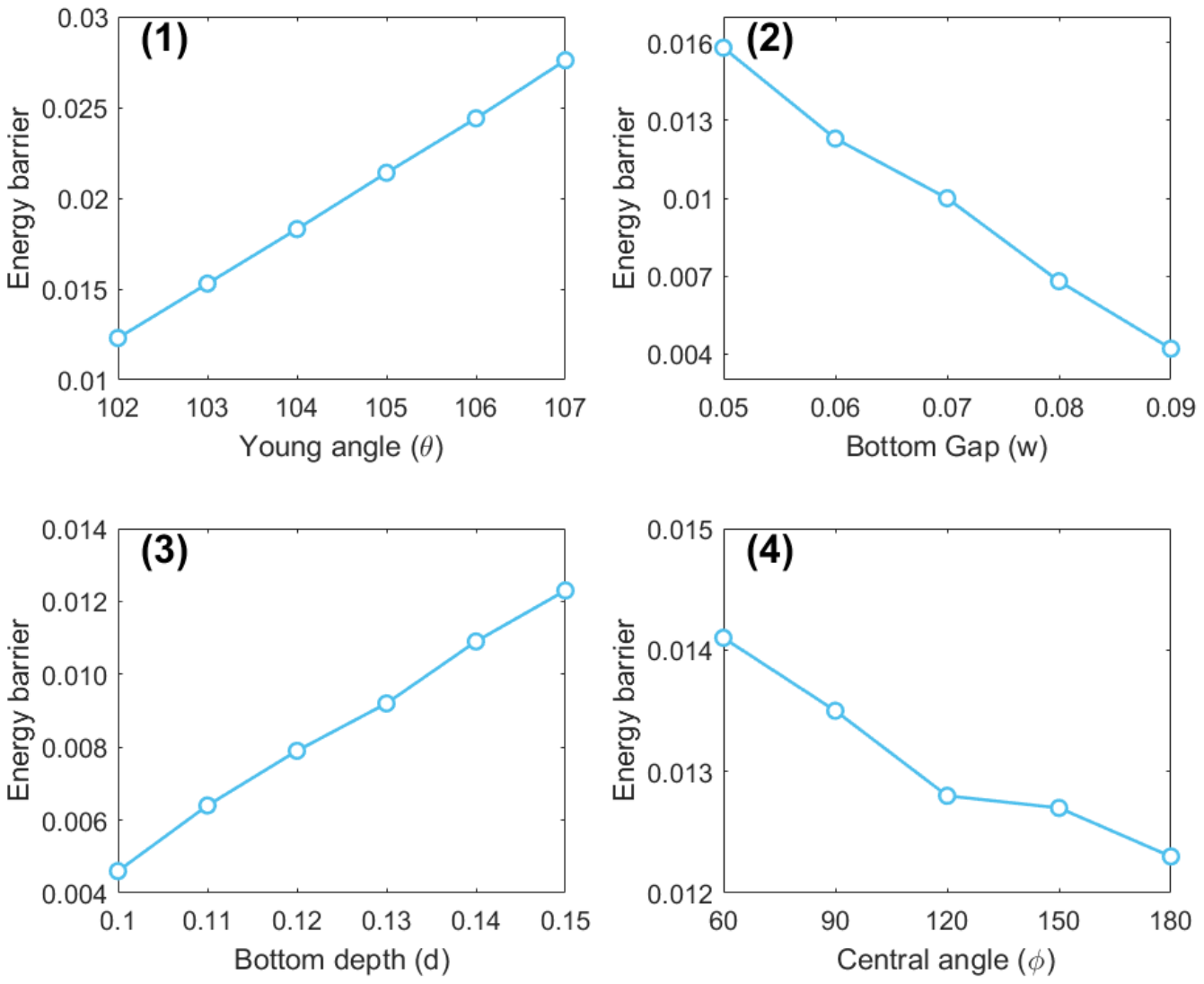}
    \centerline{(a)}
  \end{minipage}\hfill
  \begin{minipage}[b]{0.3\textwidth}
    \centering
    \includegraphics[width=\textwidth]{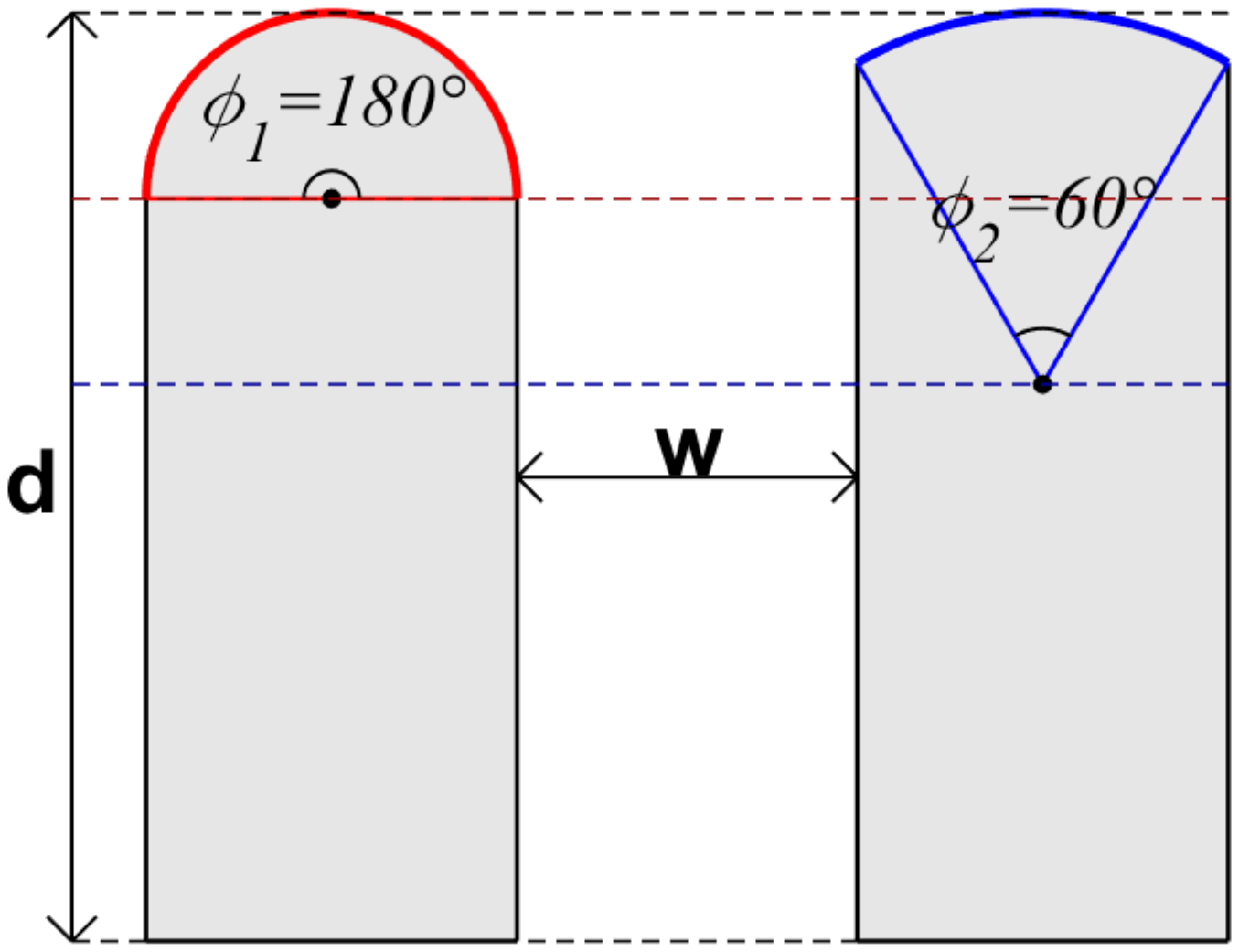}
    \vspace{0.4cm} 
    \centerline{(b)}
  \end{minipage}
  \caption{(a) Systematic parametric study of energy barriers under four distinct variations: (1) varying $\theta$ (with $w=0.06, d=0.15, \phi=180^\circ$); (2) varying $w$ (with $\theta=102^\circ, d=0.15, \phi=180^\circ$); (3) varying $d$ (with $\theta=102^\circ, w=0.06, \phi=180^\circ$); and (4) varying $\phi$ (with $\theta=102^\circ, d=0.15, w=0.06$). (b) Schematic diagram of the micro-structured substrate, demonstrating the definitions of the pillar spacing $w$, the total pillar depth $d$, and the central angles (e.g., $\phi_1$ and $\phi_2$) which govern the curvature of the convex pillar caps.}
 \label{fig:parameter_effects}
\end{figure}
\subsection{3D Wetting Transitions and Solution Landscapes}

We next examine the applicability of DD-ASD to 3D wetting transitions on textured substrates. In the first example, the interfacial width, grid size, momentum parameter, and time step are set to $\epsilon=0.015$, $h=0.015$, $\gamma=0.5$, and $\Delta t=0.005$, respectively. As shown in Figure~\ref{fig:3d_pathway}, DD-ASD identifies an index-1 saddle point 1S and the associated transition pathway connecting the metastable Cassie--Baxter (CB) and Wenzel (W) states. In this case,
the CB-to-W transition is described by a single index-1 saddle point,
providing a simple 3D transition pathway. The computation
requires $172.303$\,s, indicating that the 3D saddle-point search can be performed at a practical computational cost.

\begin{figure}[htbp]
    \centering
    \includegraphics[width=1\textwidth, trim=0cm 23.5cm 0cm 0cm, clip]{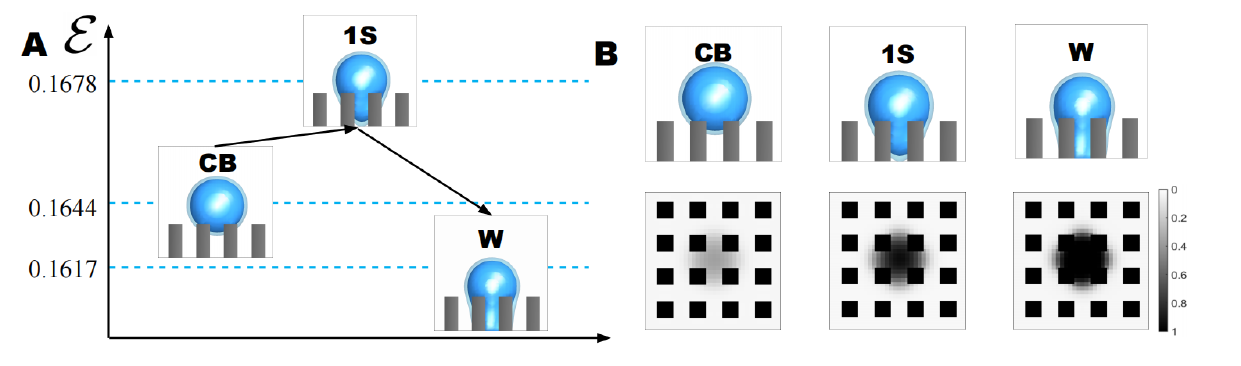}
    \caption{
3D wetting transition pathway from the Cassie--Baxter (CB) state to the Wenzel (W) state through an index-1 saddle point (1S).
Panel (A) shows the corresponding energy profile along the transition pathway.
In panel (B), the three columns represent the CB state, the transition state 1S, and the metastable W state, respectively.
The upper row shows side views in the $xz$-plane, while the lower row shows top views of the substrate; the grayscale indicates the average liquid fraction inside the grooves.
The parameters are $\epsilon=0.015$, pillar depth $d=0.15$, pillar gap $w=0.06$, and Young's angle $\theta=102^\circ$.
}
     \label{fig:3d_pathway}
\end{figure}

We then change the wetting parameters, in particular the Young's contact
angle, so that the metastable CB configuration and the subsequent
infiltration process become more complex. Under these conditions, the
transition from CB to W can no longer be described by a direct
transition pathway, but proceeds through an intermediate metastable
wetting state. This motivates the computation of the full
3D solution landscape.
\begin{figure}[htbp!]
    \centering
    \includegraphics[width=0.8\textwidth, trim=0.5cm 19cm 0.5cm 0cm, clip]{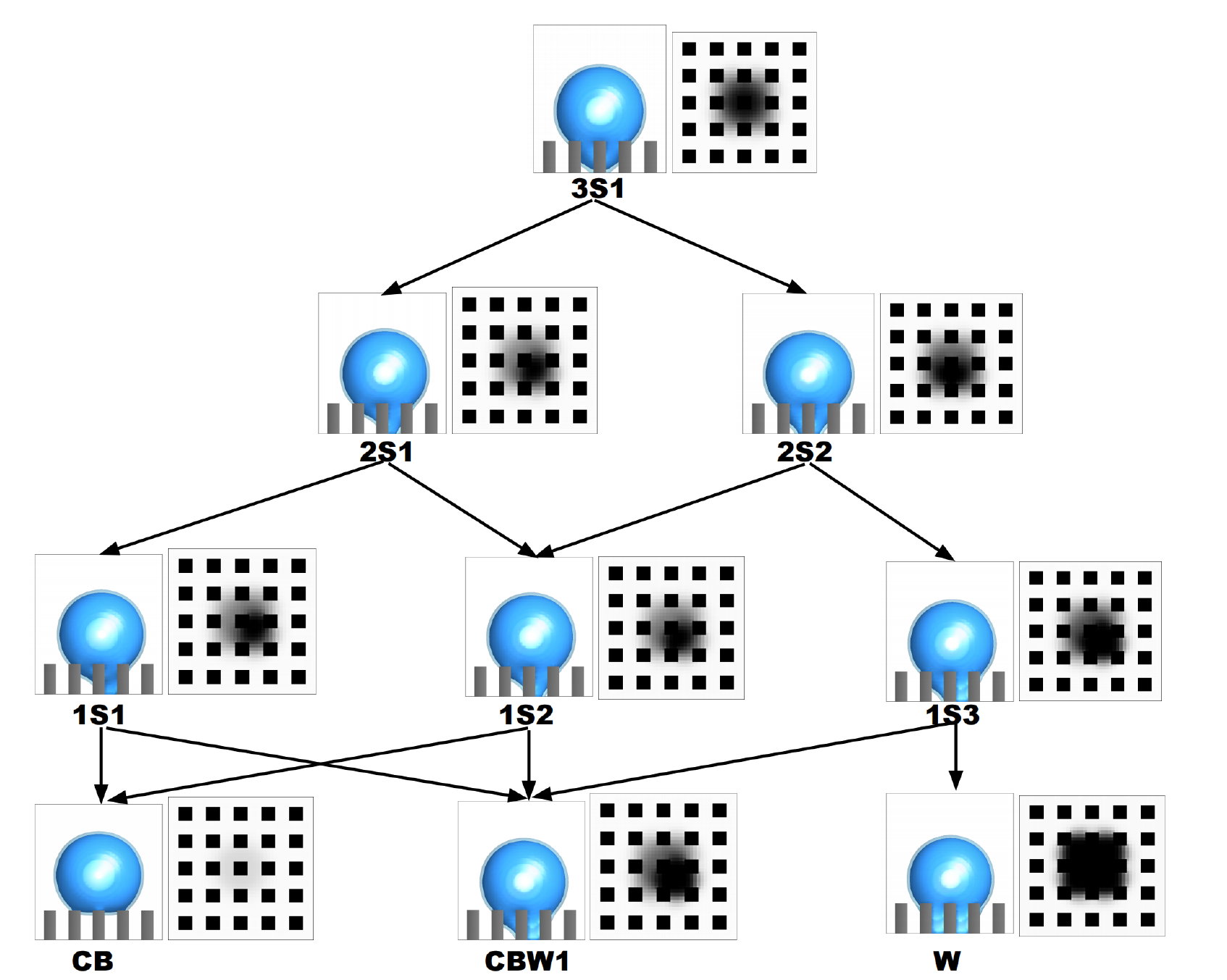}
    \caption{
3D wetting solution landscape for the nine-pillar configuration.
The landscape contains three local minima (CB, CBW1, and W), three index-1 saddle points (1S1, 1S2, and 1S3), two index-2 saddle points (2S1 and 2S2), and one index-3 saddle point (3S1).
For each stationary state, the left image shows the side view in the $xz$-plane and the right image shows the top view of the substrate; the grayscale indicates the average liquid fraction inside the grooves.
Arrows denote downward connections between stationary states of successive Morse indices.
The parameters are $\epsilon=0.015$, pillar gap $0.06$, pillar depth $0.15$, and $\theta=106^\circ$.
}
     \label{fig:3d_solution_landscape}
\end{figure}

To resolve the resulting multiple-step wetting transition mechanism, we construct the solution landscape for a wetting configuration involving nine central pillars. As shown in Figure~\ref{fig:3d_solution_landscape}, the landscape contains three local minima (CB, CBW1, and W), three index-1 saddle points (1S1, 1S2, and 1S3), two index-2 saddle points (2S1 and 2S2), and one index-3 saddle point (3S1). Their connections reveal multiple competing infiltration pathways among the metastable wetting configurations.

\section{Conclusion}

This work develops an efficient DD-ASD method for mass-constrained wetting
transitions on textured substrates. The solid geometry is represented
on a Cartesian grid, and the mass constraint is enforced through an
orthogonal projection applied to both the state and unstable-direction
dynamics. The fully discrete scheme preserves the prescribed discrete
mass and keeps the search directions in the mass-conserving subspace in exact arithmetic. After the nonphysical null modes induced by the domain mask are removed, the local convergence analysis is carried out on the effective mass-conserving subspace. Under the stated spectral and exact-eigenspace assumptions, the momentum-accelerated state iteration has a local convergence rate of $1-\mathcal{O}(1/\sqrt{\kappa})$. The numerical experiments include comparison with the reference
2D wetting landscape, step-size self-convergence, discrete mass conservation, and mesh-sensitivity tests. The efficiency of the DD-ASD method is further demonstrated by applying to 3D wetting computations.

Future work will extend the present framework to 3D
directional wetting transport on complex textured surfaces, as motivated by recent studies of geometry-controlled liquid steering \cite{feng2021capillary,dai2026nc}.
Coupling the energy landscape with hydrodynamic effects can further connect the computed transition states and barriers with actual droplet motion.

\section*{Data and Code Availability}
The data and MATLAB source code used for running the numerical experiments, performing the saddle-point computations, and generating the plots in this manuscript are openly available in a GitHub repository at \url{https://github.com/Zhang2532/DDM-ASD.git}.

\end{document}